\documentclass[a4paper,USenglish]{lipics-v2021}

\usepackage{algorithm}
\usepackage[noend]{algpseudocode}
\usepackage{caption}
\usepackage{subcaption}

\newcommand{\restate}[2]{\newtheorem*{restate-#1}{\autoref{#1}}\begin{restate-#1}#2\end{restate-#1}}

\newcommand{\dist}{\operatorname{dist}}
\def\weight{\mathbf{w}}

\title{Distributed Construction of Lightweight Spanners for Unit Ball Graphs}

\author{David Eppstein}
    {Department of Computer Science, University of California, Irvine}
    {eppstein@uci.edu}
    {}
    {}

\author{Hadi Khodabandeh}
    {Department of Computer Science, University of California, Irvine}
    {khodabah@uci.edu}
    {https://orcid.org/0000-0003-3850-6739}
    {}

\authorrunning{D. Eppstein and H. Khodabandeh}
\Copyright{David Eppstein and Hadi Khodabandeh}
\authorrunning{Eppstein and Khodabandeh}

\EventEditors{Christian Scheideler}
\EventNoEds{1}
\EventLongTitle{36th International Symposium on Distributed Computing (DISC 2022)}
\EventShortTitle{DISC 2022}
\EventAcronym{DISC}
\EventYear{2022}
\EventDate{October 25--27, 2022}
\EventLocation{Augusta, Georgia, USA}
\EventLogo{}
\SeriesVolume{246}
\ArticleNo{33}

\ccsdesc[500]{Theory of computation~Distributed algorithms}
\ccsdesc[500]{Theory of computation~Sparsification and spanners}
\ccsdesc[300]{Networks~Network control algorithms}

\keywords{spanners, doubling metrics, distributed, topology control}
\nolinenumbers

\begin{document}

\maketitle

\begin{abstract}
Resolving an open question from 2006 \cite{damian2006local}, we prove the existence of light-weight bounded-degree spanners for unit ball graphs in the metrics of bounded doubling dimension, and we design a simple $\mathcal{O}(\log^*n)$-round distributed algorithm in the LOCAL model of computation, that given a unit ball graph $G$ with $n$ vertices and a positive constant $\epsilon < 1$ finds a $(1+\epsilon)$-spanner with constant bounds on its maximum degree and its lightness using only 2-hop neighborhood information.
This immediately improves the best prior lightness bound, the algorithm of Damian, Pandit, and Pemmaraju \cite{damian2006distributed}, which runs in $\mathcal{O}(\log^*n)$ rounds in the LOCAL model, but has a $\mathcal{O}(\log \Delta)$ bound on its lightness, where $\Delta$ is the ratio of the length of the longest edge to the length of the shortest edge in the unit ball graph. Next, we adjust our algorithm to work in the CONGEST model, without changing its round complexity, hence proposing the first spanner construction for unit ball graphs in the CONGEST model of computation. We further study the problem in the two dimensional Euclidean plane and we provide a construction with similar properties that has a constant average number of edge intersections per node. Lastly, we provide experimental results that confirm our theoretical bounds, and show an efficient performance from our distributed algorithm compared to the best known centralized construction.
\end{abstract}



\maketitle


\section{Introduction}
Given a collection of points $V$ in a metric space with doubling dimension $d$, the weighted \emph{unit ball graph (UBG)} on $V$ is defined as a weighted graph $G(V, E)$ where two points $u, v\in V$ are connected if and only if their metric distance $\lVert uv\rVert\leq 1$. The weight of the edge $uv$ of the UBG is $\lVert uv\rVert$ if the edge exists. Unit ball graphs in the Euclidean plane are called unit disk graphs (UDGs) and are frequently used to model ad-hoc wireless communication networks, where every node in the network has an effective communication range $R$, and two nodes are able to communicate if they are within a distance $R$ of each other.

\emph{Spanners} are sub-graphs of the input graph whose pair-wise distances approximate distances in the input graphs, while having fewer edges than complete graphs. Given a weighted graph $G$, a $t$-spanner on $G$ can be defined as a graph $S$ that has $V(G)$ as its set of vertices, while $E(S)\subseteq E(G)$ and the following inequality is satisfied for every pair of vertices $u,v\in V(G)$:
$$\dist_S(u,v)\leq t\cdot \dist_G(u,v)$$
where $\dist_S(u,v)$ (or $\dist_G(u,v)$) is the length of the shortest path between $u$ and $v$ using the edges in $S$ (or $G$, respectively). We call this inequality the \emph{bounded stretch property}. Because of this inequality, $t$-spanners provide a $t$-approximation for the pairwise distances between the vertices in $G$. The parameter $t>1$ is called the \emph{stretch factor} or \emph{spanning ratio} of the spanner and determines how accurate the approximate distances are; spanners having smaller stretch factors are more accurate.

Spanners can be specifically defined on any graph coming from a metric space, where a heavy or undesirable network is given and finding a sparse and light-weight spanner and working with it instead of the actual network makes the computation easier and faster (Figure \ref{fig:greedy}). In particular, lightweight spanners have been gained extreme attention in the geometric setting and in the metrics with bounded doubling dimension \cite{gottlieb2015light,borradaile2019greedy,bhore2022online}, which is a generalization of the former. The problem of finding sparse light-weight spanners in these spaces has appeared in many areas of computer science, including communication network design and distributed computing. These subgraphs have few edges and are easy to construct,
leading them to appear in a wide range of applications since they were introduced \cite{chew1989planar,keil1988approximating,peleg1989graph}.  In wireless ad hoc networks $t$-spanners are used to design sparse networks with guaranteed connectivity and guaranteed bounds on routing length \cite{alzoubi2003geometric}. In distributed computing spanners provide communication-efficiency and time-efficiency through the sparsity and the bounded stretch property \cite{baswana2010additive,elkin20041e,awerbuch1998near,elkin2006efficient}. There has also been extensive use of geometric spanners in the analysis of road networks \cite{eppstein1999spanning,abam2009region,chechik2010fault}. In robotics, geometric spanners helped motion planners to design near-optimal plans on a sparse and light subgraph of the actual network \cite{dobson2014sparse,marble2013asymptotically,das1997visibility}. Spanners have many other applications including computing almost shortest paths \cite{elkin2005computing,cohen1998fast,roditty2004dynamic,feigenbaum2005graph}, and overlay networks \cite{braynard2002opus,wang2005network,jia2003local}.

\begin{figure}[ht]
     \centering
     \begin{subfigure}[b]{0.23\textwidth}
         \centering
         \includegraphics[width=\textwidth]{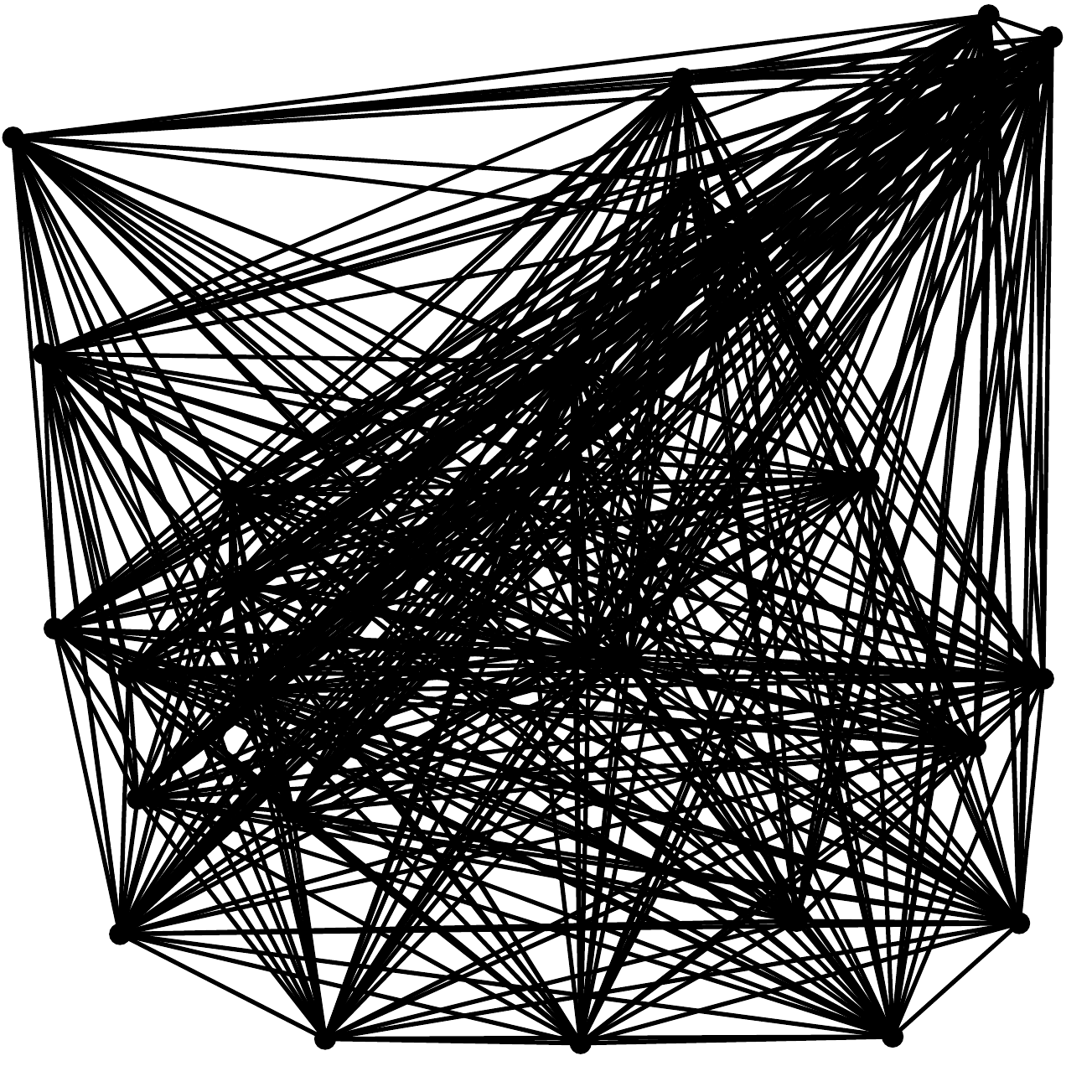}
         \caption{Complete graph}
     \end{subfigure}
     \hfill
     \begin{subfigure}[b]{0.23\textwidth}
         \centering
         \includegraphics[width=\textwidth]{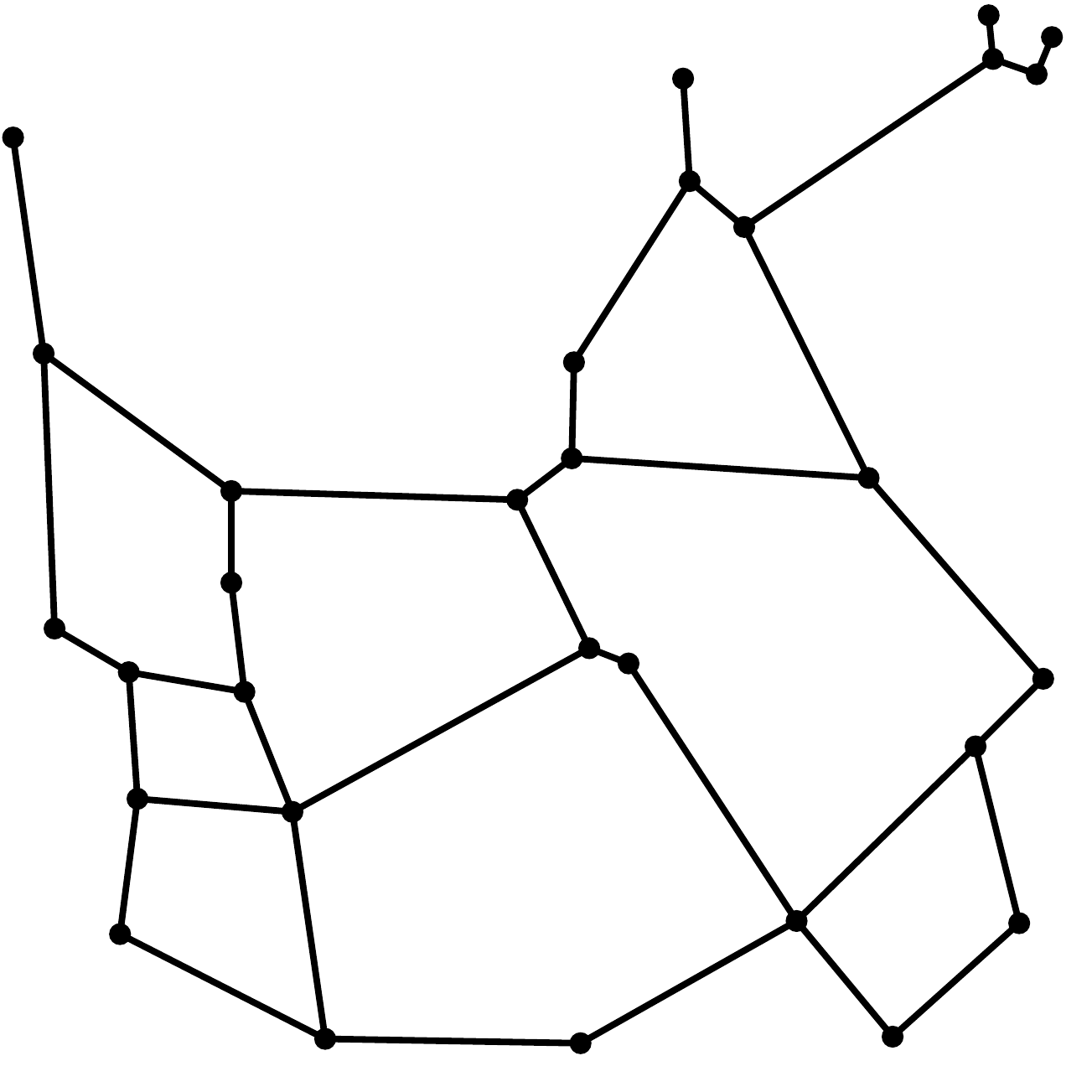}
         \caption{$2$-spanner}
     \end{subfigure}
     \hfill
     \begin{subfigure}[b]{0.23\textwidth}
         \centering
         \includegraphics[width=\textwidth]{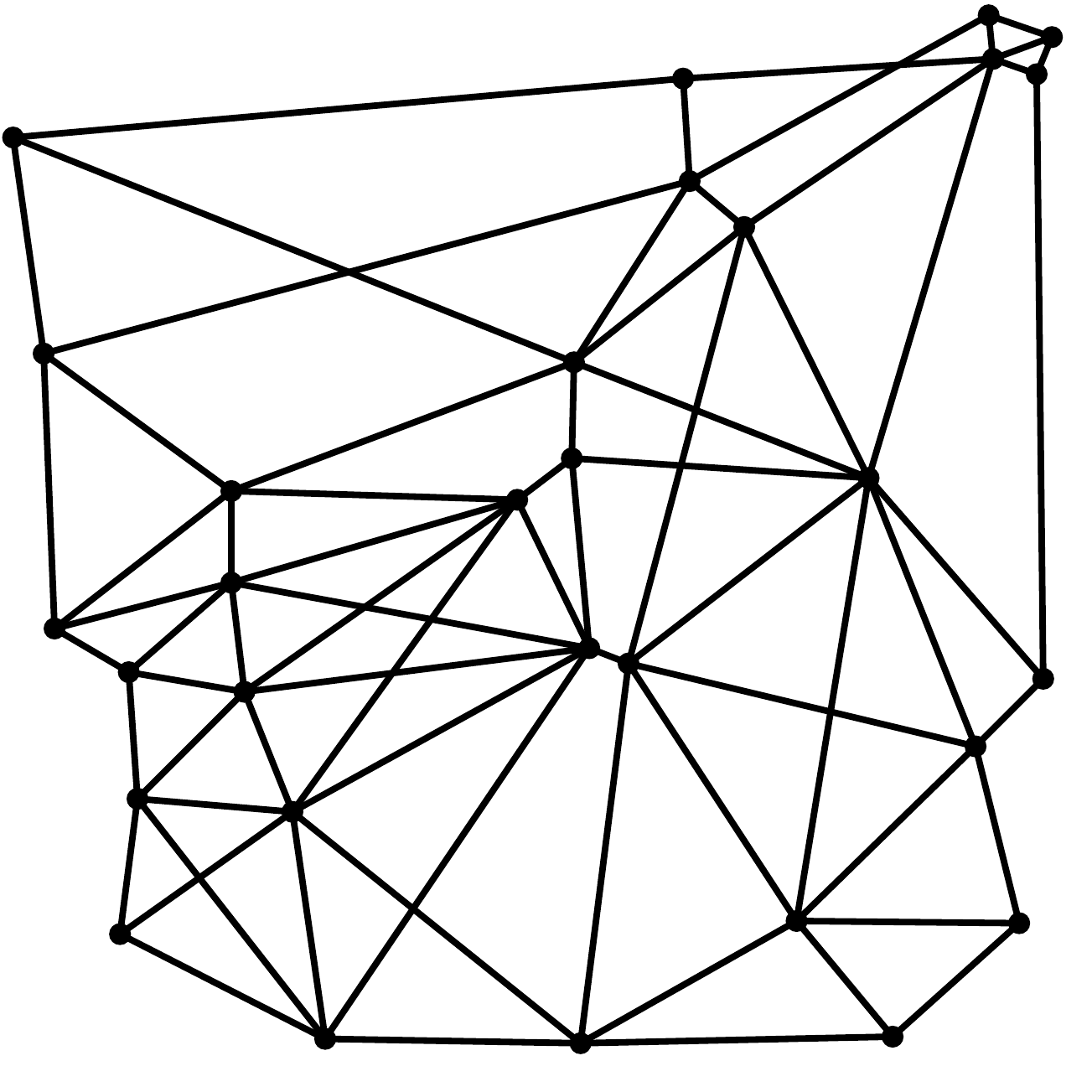}
         \caption{$1.2$-spanner}
     \end{subfigure}
     \begin{subfigure}[b]{0.23\textwidth}
         \centering
         \includegraphics[width=\textwidth]{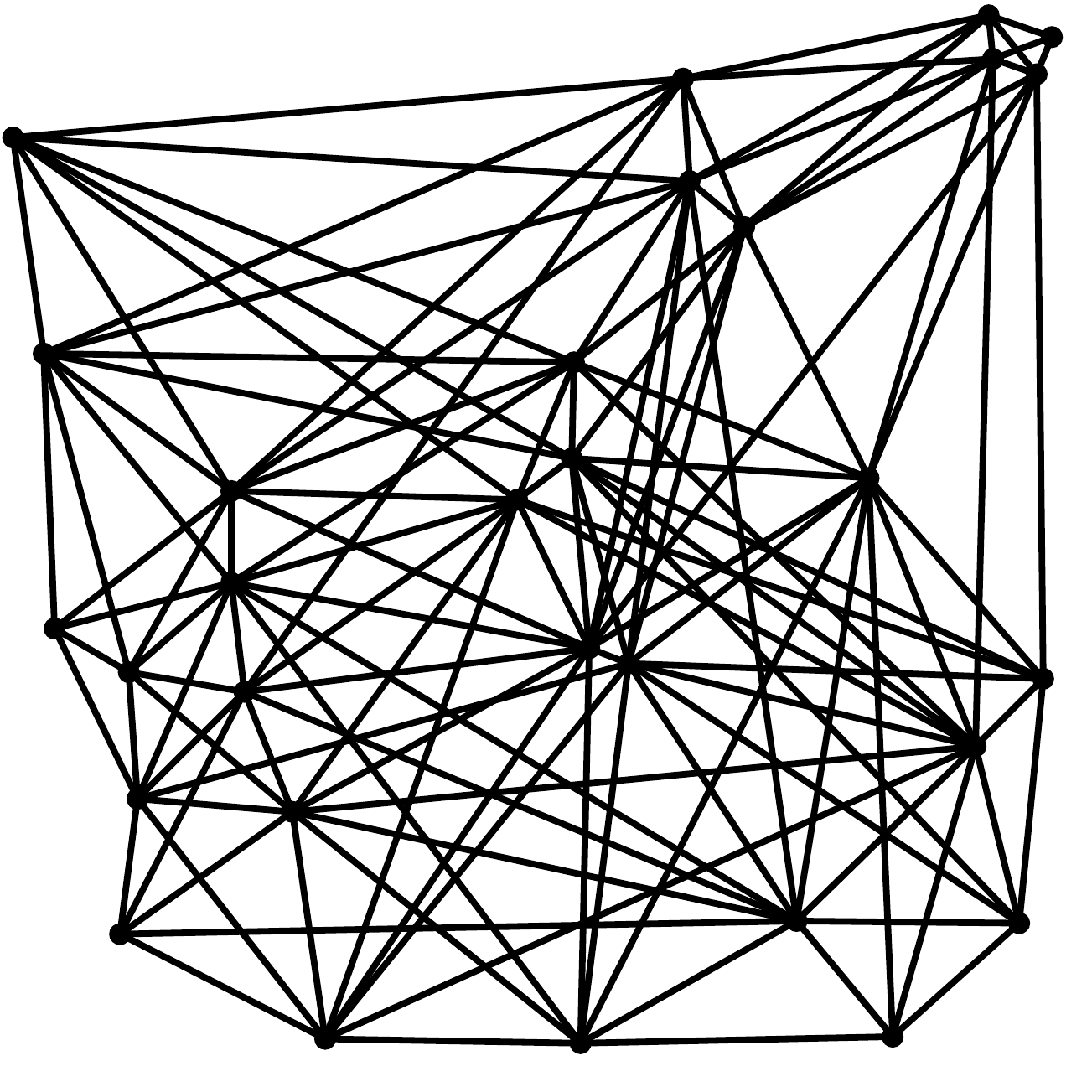}
         \caption{$1.05$-spanner}
     \end{subfigure}
    \caption{A comparison of the complete graph on $30$ random points on the plane with spanners of stretch $2$, $1.2$, and $1.05$ on the same point set.}
    \label{fig:greedy}
\end{figure}

The special case where the underlying graph is a unit ball graph is motivated by the application of unit ball graphs in modeling wireless and ad-hoc networks, where the communication of the nodes are limited by their physical distances. The problem of finding sparse lightweight spanners for unit ball graphs in this settings translates into efficient topology control algorithms. Thus the necessity of a connected and energy-efficient topology for high-level routing protocols led researchers to develop many spanning algorithms for ad-hoc networks and in particular, UDGs. But the decentralized nature of ad-hoc networks demands that these algorithms be local instead of centralized. In these applications, it is important that the resulting topology is connected, has a low weight, and has a bounded degree, implying also that the number of edges is linear in the number of vertices.

Several known \emph{proximity graphs} have been studied for this purpose, including the relative neighborhood graph (RNG), Gabriel graph (GG), Delaunay graph (DG), and Yao graph (YG). It is well-known that these proximity graphs are sparse and they can be calculated locally, using only the information from a node's neighborhood. But further analysis shows that they have poor bounds on at least one of the important criteria: maximum vertex degree, total weight, and stretch-factor \cite{li2001power}.

Researchers have modified these constructions to fulfill the requirements. Li, Wan, and Wang \cite{li2001power} introduced a modified version of the Yao graph to resolve the issue of unbounded in-degree while preserving a small stretch-factor, but they left as an open question whether there exists a construction whose total weight is also bounded by a constant factor of the weight of the minimum spanning tree. The localized Delaunay triangulation (LDT) \cite{li2003localized} and local minimum spanning tree (LMST) \cite{li2005design} were two other efforts in this way which failed to bound the total weight of the spanner. Hence bounding the weight became the main challenge in designing efficient spanners. The commonly used measure for the weight of the spanners is \emph{lightness}, which is defined as the weight of the spanner divided by the weight of the minimum spanning tree.

In the distributed setting in particular, Gao, Guibas, Hershberger, Zhang, and Zhu \cite{gao2005geometric} introduced restricted Delaunay graph (RDG), a planar distributed spanner construction for unit disk graphs in the two dimensional Euclidean plane that possessed a constant stretch-factor, leaving the weight of the spanner unstudied. Later Kanj, Perkovi{\'c}, and Xia  \cite{kanj2008computing} presented the first local spanner construction for unit disk graphs in the two dimensional Euclidean plane, which also was planar and had constant bounds on its stretch-factor, maximum degree, and lightness. Their construction was also based on the Delaunay triangulation of the point set and required information from $k$-th hop neighbors of every node, for some constant $k$ that depended on the input parameters.

In 2006, Damian, Pandit, and Pemmaraju \cite{damian2006local} designed a distributed construction for $(1+\epsilon)$-spanners of the UBGs lying in \emph{$d$-dimensional Euclidean space}. Their algorithm ran in $\mathcal{O}(\log^* n)$ rounds of communication and produced a $(1+\epsilon)$-spanner with constant bounds on its maximum degree and lightness. They used the so-called \emph{leapfrog property} to prove the constant bound on the lightness of the spanner, which does not hold for the spaces of bounded doubling dimension in general. Instead, they showed in another work \cite{damian2006distributed} that the weight of their spanner in the spaces of bounded doubling dimension is bounded by a factor $\mathcal{O}(\log \Delta)$ of the weight of the minimum spanning tree, where $\Delta$ is the ratio of the length of the longest edge in the unit ball graph divided by the length of its shortest edge. Besides these, their algorithm requires the knowledge of $\mathcal{O}(\frac{1}{\alpha-1})$-hop neighborhood of the nodes, which is costly in the CONGEST model of distributed computing, the more accepted and practical model than the LOCAL model of computation.

In the 3D Euclidean space, Jenkins, Kanj, Xia, and Zhang \cite{jenkins2012local} designed the first localized bounded-degree $(1+\epsilon)$-spanner for unit ball graphs. They also presented a lightweight construction which possessed constant bounds on its stretch-factor and maximum degree. These algorithms again required $k$-th hop neighborhood information for every node, for a constant $k$ that depended on the input parameters. Although these constructions were local, i.e. they ran in constant rounds of communication, they relied heavily on Euclidean transformations which made them inapplicable for other metric spaces.

Finally, Elkin, Filtser, and Neiman \cite{elkin2020distributed} studied the topic of lightweight spanners for general graphs and doubling graphs in the CONGEST model of distribution. For general graphs, they presented $(2k-1)\cdot(1+\epsilon)$-spanners with lightness $\mathcal{O}(k\cdot n^{1/k})$ in $\tilde{\mathcal{O}}(n^{0.5+1/(4k+2)} + D)$ rounds, where $n$ is the number of vertices and $D$ is the hop-diameter of the graph. For doubling graphs, they presented a $(1+\epsilon)$-spanner with lightness $\epsilon^{-\mathcal{O}(1)}\log n$ in $(\sqrt{n}+D)\cdot n^{o(1)}$ rounds of communication. Although these constructions are more general than the constructions of \cite{damian2006distributed} and they perform in a more restricted model (CONGEST), they do not imply a superior result in the specific case of unit ball graphs in doubling metrics.

Apart from being a generalization of the Euclidean space, the importance of the spaces of bounded doubling dimension comes from the fact that a small perturbation in the pairwise distances does not affect the doubling dimension of the point set by much, while it can change their Euclidean dimension significantly, or the resulting distances might not even be embeddable in Euclidean metrics at all \cite{chan2009small}. This makes these metrics of bounded doubling dimension to be more applicable in real-world scenarios. On the other hand, geometric arguments are considered as a strong tool for proofs of sparsity and lightness bounds in Euclidean spaces, but in doubling spaces the only available tool besides metric properties, is the packing argument which is directly followed from the definition of the doubling dimension. Therefore, the sparsity and lightness results are more difficult to achieve in the spaces of bounded doubling dimension.

Since the work of Damian, Pandit, and Pemmaraju \cite{damian2006distributed} in 2006, it has remained open whether UBGs in the spaces of bounded doubling dimension possess lightweight bounded-degree $(1+\epsilon)$-spanners and whether they can be found efficiently in a distributed model of computation. On the other hand, the construction of \cite{damian2006distributed} requires complete information about the nodes in $\mathcal{O}(\frac{1}{\alpha-1})$ hops away, for some constant $\alpha$. Acquiring this information is costly in the CONGEST model of computation, which is a more accepted model in distributed computing. Therefore, another open question arising from this line of work is to study the round complexity of the aforementioned problem in the CONGEST model. In this paper, we resolve both of these long-standing open questions by presenting centralized and distributed algorithms, both in the LOCAL, and the CONGEST model, for the purpose of finding such spanners.


\def\thmCentralized{Given a weighted unit ball graph $G$ in a metric of bounded doubling dimension and a constant $\epsilon > 0$, the spanner returned by \Call{Centralized-Spanner}{$G$,$\epsilon$} is a $(1+\epsilon)$-spanner of $G$ and has constant bounds on its lightness and maximum degree. These constant bounds only depend on $\epsilon$ and the doubling dimension.}

\def\thmDistributed{Given a weighted unit ball graph $G$ with $n$ vertices in a metric of bounded doubling dimension and a constant $\epsilon > 0$, the algorithm \Call{Distributed-Spanner}{$G$,$\epsilon$} runs in $\mathcal{O}(\log^* n)$ rounds of communication in the LOCAL model of computation, and returns a $(1+\epsilon)$-spanner of $G$ that has constant bounds on its lightness and maximum degree. These constant bounds only depend on $\epsilon$ and the doubling dimension.}

\def\thmCONGEST{Given a weighted unit ball graph $G$ with $n$ vertices in a metric of bounded doubling dimension and a constant $\epsilon > 0$, the algorithm \Call{CONGEST-Spanner}{$G$,$\epsilon$} runs in $\mathcal{O}(\log^* n)$ rounds of communication in the CONGEST model of computation, and returns a $(1+\epsilon)$-spanner of $G$ that has constant bounds on its lightness and maximum degree. These constant bounds only depend on $\epsilon$ and the doubling dimension.}

\def\thmEuclideanCentralized{Given a weighted unit disk graph $G$ in the two dimensional Euclidean plane and a constant $\epsilon > 0$, the spanner returned by \Call{Centralized-Euclidean-Spanner}{$G$,$\epsilon$} is a $(1+\epsilon)$-spanner of $G$ and has constant bounds on its lightness, maximum degree, and the average number of edge intersections per node. These constant bounds only depend on $\epsilon$ and the doubling dimension.}

\def\thmEuclideanDistributed{Given a weighted unit disk graph $G$ with $n$ vertices in the two dimensional Euclidean plane and a constant $\epsilon > 0$, the algorithm \Call{Distributed-Euclidean}{$G$,$\epsilon$} runs in $\mathcal{O}(\log^* n)$ rounds of communication and returns a bounded-degree $(1+\epsilon)$-spanner of $G$ that has constant bounds on its lightness, maximum degree, and the average number of edge intersections per node. These constant bounds only depend on $\epsilon$ and the doubling dimension.}


\subsection{Contributions}

We have two main contributions in this paper. First, we resolve the proposed open question that has remained open for more than a decade, and we prove the existence of light-weight bounded-degree $(1+\epsilon)$-spanners of unit ball graphs in the spaces of bounded doubling dimension. Our construction has constant bounds on its maximum degree and its lightness, and it can be built in $\mathcal{O}(\log^* n)$ rounds of communication in the LOCAL model of computation, where $n$ is the number of vertices.

Second, we propose the first lightweight spanner construction for unit ball graphs in the CONGEST model of computation. Even if we restrict our scope to the two dimensional Euclidean plane, where we see most of the applications of unit disk graphs, prior to this work there was no known CONGEST algorithm for finding light spanners of unit disk graphs. We achieve this construction by making adjustments on our construction for the LOCAL model to make it work in the CONGEST model in the same asymptotic number of rounds. The bounds on the lightness and maximum degree of our spanner remain the same in this model.

Besides these main results, we modify these constructions for the two dimensional Euclidean plane in order to have a linear number of edge intersections in total, implying a constant average number of edge intersections per node. This is motivated by the observation that a higher intersection per edge causes a higher chance of interference between the corresponding endpoints. To the best of our knowledge, this is the first distributed low-stretch low-intersection spanner construction for unit disk graphs.

A more detailed version of our results can be found in the following theorems. First, we introduce a centralized algorithm \Call{Centralized-Spanner}{} that,

\restate{thm:cen}{\thmCentralized} 

We use this centralized construction to propose the distributed construction \Call{Distributed-Spanner}{} in the LOCAL model of computation,

\restate{thm:dis}{\thmDistributed} 

Next, we study the problem in the CONGEST model of computation. Our distributed construction \Call{Distributed-Spanner}{} requires complete information about 2-hop neighborhood of a selected set of vertices, which is not easy to acquire in the CONGEST model. The same issues exists in the distributed algorithm of \cite{damian2006distributed}, where they aggregate information about the nodes that are $\mathcal{O}(\frac{1}{\alpha-1})$ hops away, for some constant $\alpha$. A simple approach for aggregating 2-hop neighborhoods would require $\mathcal{O}(d)$ rounds of communication in the CONGEST model, which can be as large as $\Omega(n)$ if the input graph is dense. In our next theorem, we break this barrier by making some adjustments for our algorithm to work in the CONGEST model of computation. Despite adding to the complexity of the algorithm itself, we prove that the round complexity of our new algorithm, \Call{CONGEST-Spanner}{}, would still be bounded by $\mathcal{O}(\log^*n)$.

\restate{thm:congest}{\thmCONGEST} 

Furthermore, we study the problem in the case of the two dimensional Euclidean plane, where the greedy spanner on a complete weighted graph is known to have constant upper bounds on its lightness \cite{filtser2016greedy}, maximum degree, and average number of edge intersections per node \cite{eppstein2020edge}. We observe that a simple change on the this algorithm can extend these results for unit disk graphs as well. We call this modified algorithm \Call{Centralized-Euclidean-Spanner}{} and we show that

\restate{thm:euc-cent}{\thmEuclideanCentralized} 

We use the aforementioned construction to propose \Call{Distributed-Euclidean-Spanner}{}, a specific distributed low-intersection construction for the case of the two dimensional Euclidean plane that preserves the previously mentioned properties and adds the low-intersection property.

\restate{thm:euc-dist}{\thmEuclideanDistributed} 

Besides these, we also prove that the last construction possesses sublinear separators and a separator hierarchy in the two dimensional Euclidean plane. We generalize this result to work for higher dimensions of Euclidean spaces. Finally, in section \ref{sec:exp}, we provide experimental results on random point sets in the two dimensional Euclidean plane that confirm the efficiency of our distributed construction.


\section{Preliminaries}

\subsection{Doubling metrics}

We start by recalling the definition of the doubling dimension of a metric space,

\begin{definition}[doubling dimension]
The doubling dimension of a metric space is the smallest $d$ such that for any $R>0$, any ball of radius $R$ can be covered by at most $2^d$ balls of radius $R/2$.
\end{definition}

We say a metric space has \emph{bounded doubling dimension} if its doubling dimension is upper bounded by a constant. Besides the triangle inequality, which is intrinsic to metric spaces, the \emph{packing lemma} is an essential tool for the metrics of bounded doubling dimension. This lemma states that it is impossible to pack more than a certain number of points in a ball of radius $R>0$ without making a pair of points' distance less than some $r>0$.

\begin{lemma}[Packing Property]
In a metric space of bounded doubling dimension $d$, let $X$ be a set of points with minimum distance $r$, contained in a ball of radius $R$. Then $|X|\leq \left(\frac{4R}{r}\right)^d$.

\begin{proof}
This is a well-known fact, see e.g. \cite{smid2009weak}.
\end{proof}
\end{lemma}

\subsection{Spanners for complete graphs}
For a weighted graph $G$ in a metric space, where every edge weight is equal to the metric distance of its endpoints, a $t$-spanner is defined in the following way,
\begin{definition}[$t$-spanner]
A $t$-spanner of a weighted graph $G$ is a subgraph $S$ of $G$ that for every pair of vertices $x,y$ in $G$,
$$\dist_S(x,y)\leq t\cdot \dist_G(x,y)$$
where $\dist_A(x,y)$ is the length of a shortest path between $x$ and $y$ in $A$. The lightness of $S$ is defined as $\weight(S)/\weight(MST)$ where $\weight$ is the weight function and $MST$ is the minimum spanning tree in $G$.
\end{definition}
In other words, a $t$-spanner approximates the pairwise distances within a factor of $t$. Spanners were studied for complete weighted graphs first, and several constructions were proposed to optimize them with respect to the number of edges and total weight. Among these constructions, \emph{greedy spanners} \cite{althofer1990generating} are known to out-perform the others.

A greedy spanner (\autoref{fig:greedy}) can be constructed by running the greedy spanner algorithm (\autoref{alg:greedy}) on a set of points $V$ in a metric space. This short procedure adds edges one at a time to the spanner it constructs, in ascending order by length. For each pair of vertices, in this order, it checks whether the pair already satisfies the distance inequality using the edges already added. If not, it adds a new edge connecting the pair. Therefore, by construction, each pair of vertices satisfies the inequality, either through previous edges or (if not) through the newly added edge. The resulting graph is therefore a $t$-spanner.

\begin{algorithm}[ht]
\caption{The naive greedy spanner algorithm.}\label{alg:greedy}
\begin{algorithmic}[1]
\Procedure{Naive-Greedy}{$V$}
\State Let $S$ be a graph with vertices $V$ and edges $E=\{\}$
\For {each pair $(P,Q)\in V^2$ in increasing order of $\lVert PQ\rVert$}
\If {$\dist_S(P,Q) > t\cdot \dist(P,Q)$}
\State Add edge $PQ$ to $E$%
\EndIf%
\EndFor%
\Return S%
\EndProcedure%
\end{algorithmic}
\end{algorithm}

Despite the simplicity of \autoref{alg:greedy}, Farshi and Gudmundsson~\cite{farshi2005experimental} observed that in practice, greedy spanners are surprisingly good in terms of the number of edges, weight, maximum vertex degree, and also the number of edge crossings in the two dimensional Euclidean plane. All of these properties have been proven rigorously so far. Filster and Solomon \cite{filtser2016greedy} proved that greedy spanners have size and lightness that is optimal to within a constant factor for worst-case instances. They also achieved a near-optimality result for greedy spanners in spaces of bounded doubling dimension. Borradaile, Le, and Wulff-Nilsen~\cite{borradaile2019greedy} recently proved optimality for doubling metrics, generalizing a result of Narasimhan and Smid~\cite{narasimhan2007geometric}, and resolving an open question posed by Gottlieb \cite{gottlieb2015light}, and Le and Solomon showed that no geometric $t$-spanner can do asymptotically better than the greedy spanner in terms of number of edges and lightness \cite{le2019truly}.

In a recent work, Eppstein and Khodabandeh \cite{eppstein2020edge} showed that the number of edge crossings of the greedy spanner in the two dimensional Euclidean plane is linear in the number of vertices. Moreover, they proved that the crossing graph of the greedy spanner has bounded degeneracy, implying the existence of sub-linear separators for these graphs \cite{eppstein2017crossing}. This, together with the well-known fact that greedy spanners have bounded-degree in the two dimensional Euclidean plane, makes greedy spanners more practical in this particular metric space.

Although the degree of the greedy spanner is bounded in the two dimensional Euclidean plane, it is known that there exist $n$-point metric spaces with doubling dimension 1 where the greedy spanner has maximum degree $n-1$ \cite{filtser2016greedy}. Gudmundsson, Levcopoulos, and Narasimhan \cite{gudmundsson2002fast} devised a faster algorithm that was later proven to have bounded degree as well as constant lightness and linear number of edges \cite{filtser2016greedy}. We call this algorithm \Call{Approximate-Greedy}{} in this paper, and we make use of it in our algorithms for the metrics of bounded doubling dimension, while we take advantage of the extra low-intersection property of \Call{Naive-Greedy}{} in the two dimensional Euclidean plane.

\subsection{Unit ball graphs}
We formally define a unit ball graph on a set of points $V$ in the following way,

\begin{definition}[unit ball graph]
Given a set of points $V$ in a metric space, the unit ball graph $G$ on $V$ contains $V$ as its vertex set and every two vertices $x,y\in V$ are connected if and only if $\lVert xy\rVert\leq 1$. The weight of an edge $(x,y)$ is equal to $\lVert xy\rVert$ if the edge exists.
\end{definition}

Unit ball graphs are an important subclass of the graphs called \emph{growth-bounded graphs}, which only limit the number of independent nodes in every neighborhood, a property that holds for UBGs due to the packing property.

\begin{figure}[ht]
    \centering
    \includegraphics[width=0.3\textwidth]{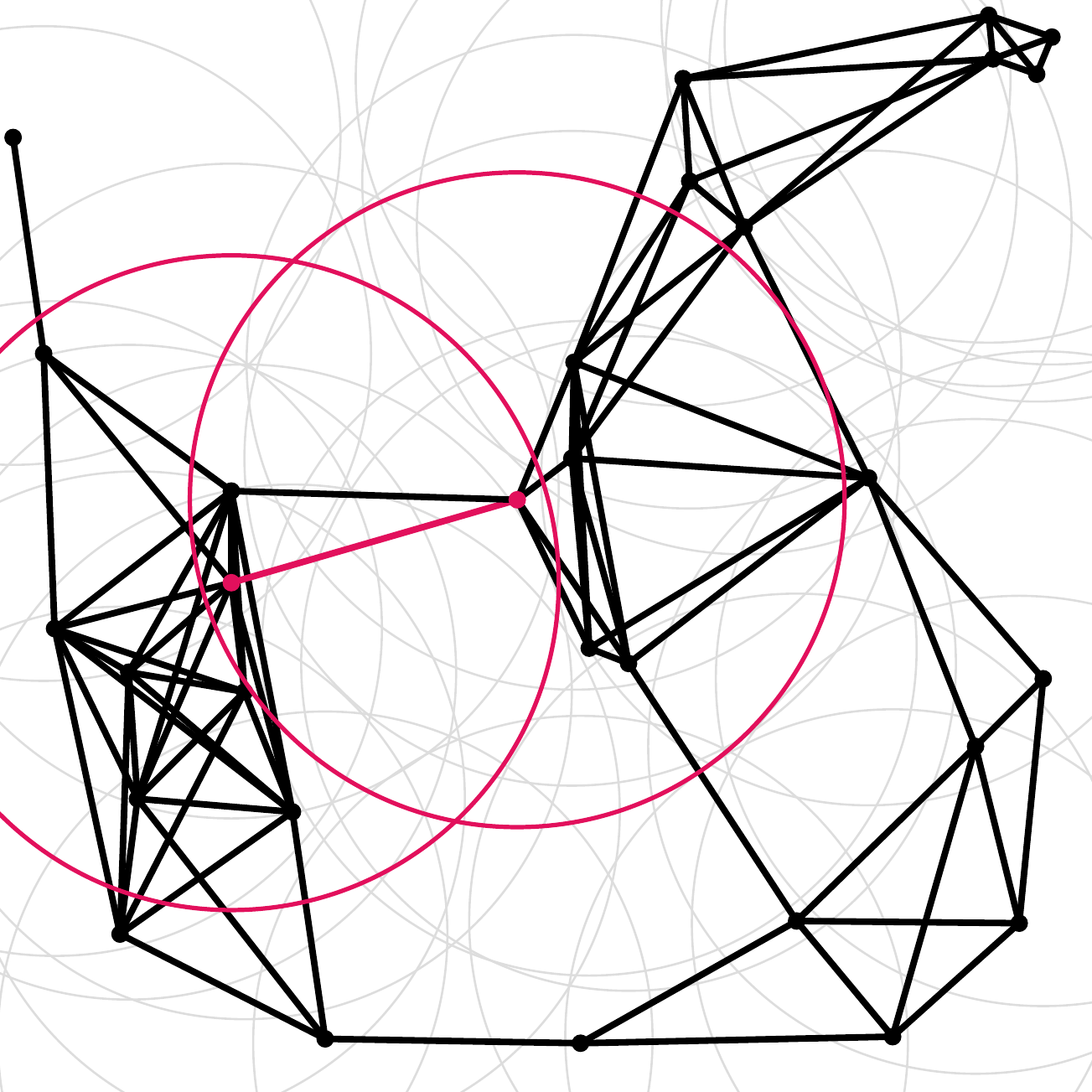}
    \caption{The unit disk graph on the same point set introduced earlier. The red disks intersect, therefore there is an edge between their centers.}
    \label{fig:ubg}
\end{figure}

Kuhn, Moscibroda, and Wattenhofer \cite{kuhn2005locality} presented a $\mathcal{O}(\log^* n)$-round distributed algorithm for finding a maximal independent set (MIS) of a unit ball graph graph in a space with bounded doubling dimension. This result was later generalized by Schneider and Wattenhofer \cite{schneider2008log} for growth-bounded graphs. Throughout the paper we refer to their algorithm by \Call{Maximal-Independent}{}. It turns out that \Call{Maximal-Independent}{} will be a key ingredient of our distributed algorithms, as well as their bottleneck in terms of the number of rounds. This means that if a maximal independent set is known beforehand, our algorithms can be executed fully locally, in constant number of rounds.

In \autoref{sec:central} we prove the existence of $(1+\epsilon)$-spanners with constant bounds on the maximum degree and the lightness by introducing an algorithm that finds such spanners in a centralized manner. In \autoref{sec:dist} we propose a distributed construction that delivers the same features through a $\mathcal{O}(\log^* n)$-round algorithm. Finally, in \autoref{sec:euc} we consider the special case of two dimensional Euclidean plane and we design centralized and distributed algorithms to construct a spanner that has the extra low-intersection property, making it more suitable for practical purposes.


\section{Centralized Construction}\label{sec:central}
In this section we propose a centralized construction for a light-weight bounded-degree $(1+\epsilon)$-spanner for unit ball graphs in a metric of bounded doubling dimension. Later in \autoref{sec:dist} we use this centralized construction to design a distributed algorithm that delivers the same features.

It is worth mentioning that the greedy spanner would be a $(1+\epsilon)$-spanner of the UBG if the algorithm stops after visiting the pairs of distance at most 1, and it even has a lightness bounded by a constant, but as we mentioned earlier, there are metrics with doubling dimension 1 in which its degree may be unbounded.

To construct a lightweight bounded-degree $(1+\epsilon)$-spanner of the unit ball graph, we start with the spanner of \cite{gudmundsson2002fast}, called \Call{Approximate-Greedy}{}, which is returns a spanner of the complete graph. It is proven in \cite{narasimhan2007geometric} that \Call{Approximate-Greedy}{} has the desired properties, i.e. bounded-degree and lightness, for complete weighted graphs in Euclidean metrics, but as stated in \cite{filtser2016greedy}, the proof only relies on the triangle inequality and packing argument which both work for doubling metrics as well. Therefore, we may safely assume that \Call{Approximate-Greedy}{} finds a light-weight bounded-degree $(1+\epsilon)$-spanner of the complete weighted graph defined on the point set. The main issue is that the edges of length more than 1 are not allowed in a spanner of the unit ball graph on the same point set. Therefore, a replacement procedure is needed to substitute these edge with edges of length at most 1. Peleg and Roditty \cite{peleg2010localized} introduced a refinement process which removes the edges of length larger than 1 from the spanner and replaces them with three smaller edges to make the output a subgraph of the UBG. The main issue with their approach is that it can lead to vertices having unbounded degrees in the spanner, therefore missing an important feature. Here, we introduce our own refinement process that not only replaces edges of larger than 1 with smaller edges and makes the spanner a subgraph of the unit ball graph, but also guarantees a constant bounded on the degrees of the resulting spanner.

\subsection{The algorithm}

In the first step of the algorithm (\autoref{alg:centralgreedy}) we choose $\epsilon'=\epsilon/36$, a smaller stretch parameter than $\epsilon$, to cover the errors that future steps might inflict to the spanner. Then we call the procedure \Call{Approximate-Greedy}{} on the set of vertices $V$ to calculate a light-weight bounded-degree $(1+\epsilon')$-spanner $S$ of the complete weighted graph on $V$. This spanner might contain edges of length larger than 1, which we will replace by some edges of length at most 1 in the future steps.

Since an edge of length larger than $1+\epsilon'$ in $S$ cannot participate in the shortest path between any two adjacent vertices in $G$, we simply remove and discard them from the spanner. Then for every remaining edge $e=(u,v)$ of length in the range $(1,1+\epsilon']$ we find an edge $(x,y)$ of the original graph $G$ so that $\lVert ux\rVert \leq \epsilon'$ and $\lVert vy\rVert \leq \epsilon'$. We then replace such an edge $e$ by the edge $(x,y)$. We call the pair $(x,y)$ the \emph{replacement edge} or the \emph{replacement pair} for the edge $e$. Since this procedure can end up assigning too many replacement edges to a single vertex ($x$ or $y$ in this case) and hence increasing its degree significantly, we perform a simple check before adding a replacement edge; we store the set $R$ of previously added replacement pairs in the memory and if a \emph{weak} replacement pair $(x',y')\in R$ exists, then we prefer it over a newly found replacement pair $(x,y)\notin R$. By weak replacement pair we mean a pair $(x',y')\in R$ that $\lVert ux'\rVert \leq 2\epsilon'$ and $\lVert vy'\rVert \leq 2\epsilon'$, which is weaker than the definition of the replacement pair. As we later see this weaker notion of replacement pair will help us to bound the degree of the vertices.

After removing the edges of length larger than 1 and replacing the ones in the range $(1,1+\epsilon']$, we return the spanner to the output.

\begin{algorithm}[ht]
\caption{A centralized spanner construction.}\label{alg:centralgreedy}
\textbf{Input.} A unit ball graph $G(V,E)$ in a metric with doubling dimension $d$.\\
\textbf{Output.} A light-weight bounded-degree $(1+\epsilon)$-spanner of $G$.%
\begin{algorithmic}[1]%
\Procedure{Centralized-Spanner}{$G$, $\epsilon$}%
\State $\epsilon'\gets \epsilon/36$
\State $S \gets$ \Call{Approximate-Greedy}{$V$, $\epsilon'$}%
\State R $\gets\varnothing$
\For {$e=(u,v)$ in $S$}%
\If {$|e| > 1$}%
\State Remove $e$ from $S$%
\EndIf%
\If {$|e|\in (1,1+\epsilon']$}%
\If {$\exists (x,y)\in E$ that $\lVert ux\rVert\leq \epsilon'$ and $\lVert vy\rVert\leq \epsilon'$}%
\If {$\nexists (x',y')\in R$ that $\lVert ux'\rVert\leq 2\epsilon'$ and $\lVert vy'\rVert\leq 2\epsilon'$}%
\State $S\gets S\cup \{(x,y)\}$%
\State $R\gets R\cup \{(x,y)\}$
\EndIf%
\EndIf%
\EndIf%
\EndFor%
\State \Return S%
\EndProcedure%
\end{algorithmic}
\end{algorithm}

\subsection{The analysis}

Now we prove that the output $S$ of the algorithm is a light-weight bounded-degree $(1+\epsilon)$-spanner of the unit ball graph $G$. Clearly, after the refinement is done the spanner $S$ is a subgraph of $G$, so we need to analyze the lightness, the stretch factor, and the maximum degree of the spanner.

First we prove that the stretch-factor of the spanner is indeed bounded by $1+\epsilon$.

\def\lemCenStr{The spanner returned by \Call{Centralized-Spanner}{} has a stretch factor of $1+\epsilon$.}%
\begin{lemma}%
\lemCenStr\label{lem:cen-str}
\end{lemma}
\begin{proof}
We recall that the output of \Call{Approximate-Greedy}{} is a light-weight bounded-degree $(1+\epsilon')$-spanner of the complete weighted graph on the point set. So an edge $e$ of length $|e|>1+\epsilon'$ cannot be used to approximate any edges in the UBG, i.e. if $(x,y)\in E$ then $e$ cannot belong to the shortest path between $x$ and $y$ in $S$; otherwise the length of the path would exceed $1+\epsilon'$ which cannot happen. So we may safely remove these edges in the first step of the refinement procedure without replacing them.

\begin{figure}[ht]
    \centering
    \includegraphics[width=0.4\textwidth]{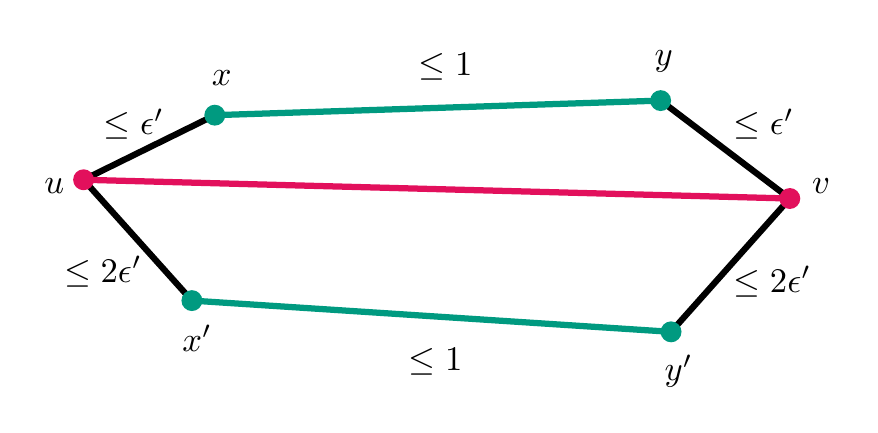}
    \caption{An edge $(x,y)$ of the UBG that uses a longer than unit length edge $(u,v)$ of the spanner on its shortest path, which is then replaced by $(x',y')$ during the replacement procedure.}
    \label{fig:rep-edge}
\end{figure}

Also, any edge of length in the range $(1,1+\epsilon']$ that is not used in a shortest path between any two endpoints of an edge of UBG can be removed as well, because removing them does not change the stretch-factor of the spanner. Now consider an edge $(x,y)\in E$ of the UBG that uses a spanner edge $e = (u,v)\in S$ that $|e|\in(1,1+\epsilon']$ on its shortest path. We want to prove that after the replacement of $e$, the shortest path between $x$ and $y$ remains within $1+\epsilon$ factor of their distance. Clearly, we have $\lVert ux\rVert\leq \epsilon'$ and $\lVert vy\rVert\leq \epsilon'$; otherwise the length of the path $xuvy$ would be more than $1+\epsilon'$, contradicting the fact that it is approximating an edge of length at most 1. This shows that $(x,y)$ would be a valid replacement edge for $e$. So we can safely assume that the algorithm finds a (possible weak) replacement edge $(x',y')\in E$ for $e$ (\autoref{fig:rep-edge}). This replacement edge might be a normal replacement edge or a weak replacement edge. Either way, we have $\lVert ux'\rVert \leq 2\epsilon'$ and $\lVert vy'\rVert\leq 2\epsilon'$. By the triangle inequality
\begin{equation*}
    \lVert x'x\rVert \leq \lVert x'u\rVert + \lVert ux\rVert \leq 2\epsilon' + \epsilon' = 3\epsilon'
\end{equation*}
Similarly, $\lVert yy'\rVert \leq 3\epsilon'$. Therefore
\begin{equation}\label{eq:cen-eq1}
    \lVert x'y'\rVert \leq \lVert x'x\rVert + \lVert xy\rVert + \lVert yy'\rVert
    \leq \lVert xy\rVert + 6\epsilon'
\end{equation}
Denote the shortest spanner path between $x$ and $x'$ by $P_{xx'}$ and similarly define $P_{yy'}$. Consider the spanner path $P=P_{xx'}x'y'P_{y'y}$ that connects $x$ and $y$. Using \autoref{eq:cen-eq1} the length of the path $P$ is
\begin{equation}\label{eq:cen-eq2}
    |P|=|P_{xx'}| + \lVert x'y'\rVert + |P_{y'y}| \leq \lVert xy\rVert + 6\epsilon' + |P_{xx'}| + |P_{y'y}|
\end{equation}
The changes that we make in the refinement process do not affect the length of short paths like $P_{xx'}$ and $P_{y'y}$. So we have
\begin{equation*}
    |P_{xx'}|\leq (1+\epsilon')\lVert xx'\rVert\leq 3\epsilon'(1+\epsilon')
\end{equation*}
Similarly, $|P_{yy'}|\leq 3\epsilon'(1+\epsilon')$. Putting these into \autoref{eq:cen-eq2} and using $\epsilon=36\epsilon'$,
\begin{equation}\label{eq:cen-eq3}
    |P|\leq \lVert xy\rVert + 6\epsilon' + 6(1+\epsilon')\epsilon' \leq \lVert xy\rVert + \frac{\epsilon}{1+\epsilon'}
\end{equation}
But since $e$ was previously approximating the edge $(x,y)$, we know that $(1+\epsilon')\lVert xy\rVert \geq |e| > 1$ or equivalently $\lVert xy\rVert > 1/(1+\epsilon')$. Substituting this into \autoref{eq:cen-eq3},
\begin{equation*}
    |P|\leq \lVert xy\rVert+\epsilon\lVert xy\rVert  = (1+\epsilon)\lVert xy\rVert
\end{equation*}
So $S$ is a $(1+\epsilon)$-spanner of $G$.
\end{proof}

Now we analyze the weight of the spanner, proving its constant lightness.

\def\lemCenLight{The spanner returned by \Call{Centralized-Spanner}{} has a weight of $\mathcal{O}(1)\weight(MST)$.}%
\begin{lemma}
\lemCenLight\label{lem:cen-light}
\end{lemma}

\begin{proof}
Again, we use the fact that the output of \Call{Approximate-Greedy}{} has weight $\mathcal{O}(1)\weight(MST(G))$. During the refinement process, every edge is replaced by an edge of smaller length, so the whole weight of the graph does not increase during the refinement process. Therefore in the end $\weight(S) = \mathcal{O}(1)\weight(MST(G))$.
\end{proof}

In the final step, we bound the maximum degree of the spanner.

\def\lemCenDeg{The spanner returned by \Call{Centralized-Spanner}{} has bounded degree.}%
\begin{lemma}
\lemCenDeg\label{lem:cen-deg}
\end{lemma}
\begin{proof}
It is clear from the algorithm that immediately after processing an edge $e=(u,v)$, the degree of $u$ and $v$ does not increase; it may decrease due to the removal of the edge which is fine. But if a replacement edge $(x,y)$ is added after the removal of $e$ then the degree of $x$ and $y$ is increased by at most 1. We need to make sure this increment is bounded for every vertex.

Let $x$ be an arbitrary vertex of $G$ and let $(x,y)$ and $(x,z)$ be two replacement edges that have been added to $x$ in this order as a result of the refinement process. We claim that $\lVert yz\rVert > \epsilon'$ holds. Assume, on the contrary, that $\lVert yz\rVert \leq \epsilon'$, and also assume that $(x,z)$ has been added in order to replace an edge $(u,v)$ of the spanner. Then by the triangle inequality
\begin{equation*}
    \lVert vy\rVert \leq \lVert vz\rVert + \lVert zy\rVert \leq 2\epsilon'
\end{equation*}
Also $\lVert ux\rVert \leq \epsilon' < 2\epsilon'$ because $(x,z)$ is added to replace $(u,v)$. But the last two inequalities contradict the fact that $(x,y)$ cannot be a weak replacement for $(u,v)$.

Now that we have proved $\lVert yz\rVert > \epsilon'$ we can use the packing property of the bounded doubling dimension to bound the number of such replacement edges around $x$. All the other endpoints of such replacement edges are included in ball of radius 1 around $x$, and the distance between every two such points is at least $\epsilon'$. Thus by the packing property there can be at most $(\frac{4}{\epsilon'})^d=\epsilon^{-\mathcal{O}(d)}$ many replacement edges incident to $x$.
\end{proof}

Putting these together, we can prove \autoref{thm:cen}.

\begin{theorem}[Centralized Spanner]
\thmCentralized\label{thm:cen}
\end{theorem}
\begin{proof}
Follows directly from Lemma \ref{lem:cen-str}, Lemma \ref{lem:cen-light}, and Lemma \ref{lem:cen-deg}.
\end{proof}


\section{Distributed Construction}\label{sec:dist}
In this section we propose our distributed construction for finding a $(1+\epsilon)$-spanner of a unit ball graph using only 2-hop neighborhood information. The spanner returned by our algorithm has constant bounds on its maximum degree and its lightness. This is the first light-weight distributed construction for unit ball graphs in doubling metrics, to the best of our knowledge.

In our distributed construction, we run our centralized algorithm on the 2-hop neighborhoods of a an independent set of the unit ball graph, and we prove that putting these local spanners together will achieve a spanner that possesses the desired properties.

\subsection{The algorithm}

For the distributed construction we propose Algorithm \ref{alg:localgreedy}. There is a preprocessing step of finding a maximal independent set $I$ of $G$, which can be done using the distributed algorithm of \cite{kuhn2005locality} in $\mathcal{O}(\log^*n)$ rounds. We refer to this algorithm by \Call{Maximal-Independent}{}. Then the \Call{Local-Greedy}{} subroutine is run on every vertex $w\in I$ to find a $(1+\epsilon)$-spanner $S_w$ of the 2-hop neighborhood of $w$, denoted by $\mathcal{N}^2(w)$. At the final step, every $w\in I$ sends its local spanner edges to the corresponding endpoints of every edge. Symmetrically, every vertex listens for the edges sent by the vertices in $I$ and once a message is received, it stores the edges in its local storage. In other words, the final spanner is the union of all these local spanners. We use the centralized algorithm of \autoref{sec:central} for every local neighborhood $\mathcal{N}^2(w)$ to guarantee the bounds that we need.

\begin{algorithm}[ht]
\caption{The localized greedy algorithm.}\label{alg:localgreedy}
\textbf{Input.} A unit ball graph $G(V, E)$ in a metric with doubling dimension $d$ and an $\epsilon>0$.\\
\textbf{Output.} A light-weight bounded-degree $(1+\epsilon)$-spanner of $G$.%
\begin{algorithmic}[1]%
\Procedure{Distributed-Spanner}{$G$, $\epsilon$}%
\State Find a maximal independent set $I$ of $G$ using \cite{kuhn2005locality}%
\State Run \Call{Local-Greedy}{} on the vertices of $G$%
\EndProcedure%
\Function{Local-Greedy}{vertex $w$}%
\State Retrieve $\mathcal{N}^2(w)$, the 2-hop neighborhood information of $w$%
\If {$w\in I$}%
\State $\mathcal{S}_w \gets $ \Call{Centralized-Spanner}{$\mathcal{N}^2(w)$, $\epsilon$}%
\For{$e=(u,v)$ in $\mathcal{S}_w$}%
\State Send $e$ to $u$ and $v$%
\EndFor%
\EndIf%
\State Listen to incoming edges and store them%
\EndFunction%
\end{algorithmic}
\end{algorithm}

Similar to the aforementioned greedy algorithm (\autoref{alg:greedy}), our algorithm seems very simple in the first sight. But as we see later in this section, proving its properties, particularly its lightness, is a non-trivial task.

\subsection{The analysis}
Now we show that the spanner introduced in Algorithm \ref{alg:localgreedy} possesses the desired properties. First, we show the round complexity of $\mathcal{O}(\log^* n)$.

\def\lemDisRnd{\Call{Distributed-Spanner}{} can be done in $\mathcal{O}(\log^*n)$ rounds of communication.}%
\begin{lemma}
\lemDisRnd\label{lem:dis-rnd}
\end{lemma}
\begin{proof}
The pre-processing step of finding the maximal independent set takes $\mathcal{O}(\log^*n)$ rounds of communication \cite{kuhn2005locality}. Retrieving the 2-hop neighborhood information can be done in $\mathcal{O}(1)$ rounds of communication. Computing the greedy spanner is done locally, and the edges are sent to their endpoints, which again can be done in $\mathcal{O}(1)$ rounds of communication. Overall, the algorithm requires $\mathcal{O}(\log^*n)$ rounds of communication.
\end{proof}

Next we bound the stretch-factor of the spanner.

\def\lemDisStr{The spanner returned by \Call{Distributed-Spanner}{} has a stretch factor of $1+\epsilon$.}%
\begin{lemma}
\lemDisStr\label{lem:dis-str}
\end{lemma}
\begin{proof}
From \autoref{sec:central} we know that $\mathcal{S}_w$ is a light-weight bounded-degree $(1+\epsilon)$-spanner of $\mathcal{N}^2(w)$. Let $u,v\in V$ be chosen arbitrarily. We need to make sure there is a path of length at most $(1+\epsilon)d_G(u,v)$ between $u$ and $v$ in the output.

First we prove this for the case that $(u,v)\in E$. So let $e=(u,v)\in E$. Then $u$ is either in $I$ or has a neighbor in $I$, according to choice of $I$. In any case, the edge $e$ belongs to $\mathcal{N}^2(w)$ for some $w\in I$, which means that there is a path $P\subset \mathcal{S}_w$ of length at most $(1+\epsilon)|e|$ that connects $u$ and $v$. The edges of $P$ are all included in the final spanner according to the algorithm, so the output includes this path between $u$ and $v$, which has a length at most $(1+\epsilon)|e|$ and so the distance inequality is satisfied.

If $(u,v)\notin E$, we can take the shortest path $u=p_0,p_1,\cdots,p_k=v$ between them in $G$ and append the corresponding $(1+\epsilon)$-approximate paths $P_0,P_1,\cdots,P_{k-1}$ of the edges $p_0p_1, p_1p_2,\cdots, p_{k-1}p_k$, respectively, to get a $(1+\epsilon)$-approximate path for $p_0p_1\cdots p_k$. This implies that the stretch factor of the output is indeed $1+\epsilon$.
\end{proof}

Now we bound the maximum degree of the spanner.

\def\lemDisDeg{The spanner returned by \Call{Distributed-Spanner}{} has a bounded degree.}%
\begin{lemma}
\lemDisDeg\label{lem:dis-deg}
\end{lemma}
\begin{proof}
First we use the packing lemma to prove that any vertex $v\in V$ appears at most a constant number of times in different neighborhoods, $\mathcal{N}^2(w)$ for $w\in I$. Because $v\in \mathcal{N}^2(w)$ implies that $\lVert vw\rVert \leq 2$, any vertex $w\in I$ such that $v\in \mathcal{N}^2(w)$ should be contained in the ball of radius 2 around $v$. But all such $w$s are chosen from $I$, which is an independent set of $G$, so the distance between every two such vertex is at least 1. By the packing property, the maximum number of such vertices would be $8^d = \mathcal{O}(1)$.

Now that every vertex appears in at most in $8^d$ different sets $\mathcal{N}^2(w)$, for $w\in I$, and from \autoref{sec:central} we already knew that every vertex has bounded degree in any of $\mathcal{S}_w$s, it immediately follows that every vertex has bounded degree in the final spanner.
\end{proof}

In order to bound the lightness of the output, we assume that $\epsilon\leq 1$ and we make a few comparisons. First, for any $w\in I$ we compare the weight of $\mathcal{S}_w$ to the weight of the minimum spanning tree on $\mathcal{N}^2(w)$. Then we compare the weight of the minimum spanning tree on $\mathcal{N}^2(w)$ to the weight of the minimum Steiner tree on $\mathcal{N}^3(w)$, where the required vertices are $\mathcal{N}^2(w)$ and 3-hop vertices are optional. Finally, we compare the weight of this minimum Steiner tree to the weight of the induced subgraph of \Call{Centralized-Spanner}{$G$, $\epsilon$} on the subset of vertices $\mathcal{N}^3(w)$, which later implies that the overall weight of $\mathcal{S}_w$s is bounded by a constant factor of the weight of the minimum spanning tree on $G$.

Our first claim is that the weight of $\mathcal{S}_w$ is bounded by a constant factor of the weight of the MST on $\mathcal{N}^2(w)$.

\begin{corollary}
$\weight(\mathcal{S}_w) = \mathcal{O}(1)\weight(MST(\mathcal{N}^2(w)))$\label{cor:dis-cmp1}
\end{corollary}
\begin{proof}
Follows from the properties of the centralized algorithm in \autoref{sec:central}.
\end{proof}

Next we compare $\weight(MST(\mathcal{N}^2(w)))$ to the weight of the minimum Steiner tree of $\mathcal{N}^3(w)$ on the required vertices $\mathcal{N}^2(w)$.

\def\lemDisCmpTwo{Define $\mathcal{T}$ to be the optimal Steiner tree on the set of vertices $\mathcal{N}^3(w)$, where only vertices in $\mathcal{N}^2(w)$ are required and the rest of them are optional. Then
$$\weight(MST(\mathcal{N}^2(w)))\leq 2\weight(\mathcal{T})$$}%
\begin{lemma}%
\label{lem:dis-cmp2}%
\lemDisCmpTwo
\end{lemma}
\begin{proof}
This is a well-known fact that implies a 2-approximation for minimum Steiner tree problem. The idea is if we run a full DFS on the vertices of $\mathcal{T}$ and we write every vertex once we open and once we close it, then we get a cycle whose shortcut on optional edges will form a path on the required vertices. The weight of the cycle is at least $\weight(MST(\mathcal{N}^2(w)))$ and at most $2\weight(\mathcal{T})$, which proves the result.
\end{proof}

We then compare the weight of $\mathcal{T}$ to the weight of induced subgraph of \Call{Centralized-Spanner}{$G$, $\epsilon$} on the subset of vertices $\mathcal{N}^3(w)$. The main observation here is that when $\epsilon\leq 1$ the induced subgraph of the centralized spanner on $\mathcal{N}^3(w)$ would be a feasible solution to the minimum Steiner tree problem on $\mathcal{N}^3(w)$, with the required vertices being the vertices in $\mathcal{N}^2(w)$. This will imply that the weight of the induced subgraph is at least equal to the weight of the minimum Steiner tree.

\def\lemDisCmpThree{Let $\mathcal{S}^*$ be the output of \Call{Centralized-Spanner}{$G$, $\epsilon$} and let $\mathcal{S}^*_w$ be the induced subgraph of $\mathcal{S}^*$ on $\mathcal{N}^3(w)$. Then
$$\weight(\mathcal{T}) \leq \weight(\mathcal{S}^*_w)$$}%
\begin{lemma}%
\label{lem:dis-cmp3}%
\lemDisCmpThree
\end{lemma}
\begin{proof}
We prove that for $\epsilon\leq 1$, $\mathcal{S}^*_w$ forms a forest that connects all the vertices in $\mathcal{N}^2(w)$ in a single component. So $\mathcal{S}^*_w$ is a feasible solution to the minimum Steiner tree problem on the set of vertices $\mathcal{N}^3(w)$ with required vertices being $\mathcal{N}^2(w)$. Thus $\weight(\mathcal{T}) \leq \weight(\mathcal{S}^*_w)$.

Now we just need to prove that the vertices in $\mathcal{N}^2(w)$ are connected in $\mathcal{S}^*_w$. Let $u$ be an $i$-hop neighbor of $w$ and $v$ be an $i+1$-hop neighbor of $w$ for some $w\in I$ and $i=0,1$. Assume that $(u,v)\in E$. It is enough to prove that $u$ and $v$ are connected in $\mathcal{S}^*_w$. In order to do so, we observe that there is a path of length at most $(1+\epsilon)\lVert uv\rVert$ between $u$ and $v$ in $\mathcal{S}^*$. We show that this path is contained in $\mathcal{N}^3(w)$ and we complete the proof in this way, because $\weight(\mathcal{S}^*_w)$ is nothing but the induced subgraph of $\mathcal{S}^*$ on $\mathcal{N}^3(w)$.

Assume, on the contrary, that there is a vertex $x\notin \mathcal{N}^3(w)$ on the $(1+\epsilon)$-path between $u$ and $v$. This means that $x$ is not a 1-hop neighbor of any of $u$ and $v$, because otherwise $x$ would have been in $\mathcal{N}^3(w)$. So $\lVert ux\rVert>1$ and $\lVert vx\rVert>1$. Thus the length of the path would be at least $\lVert ux\rVert+\lVert xv\rVert>2\geq (1+\epsilon)\geq (1+\epsilon)\lVert uv\rVert$ which is a contradiction.
\end{proof}

This lemma concludes our sequence of comparisons. By putting together what we proved so far, we have
\begin{proposition}%
\label{prop:lightness}%
The spanner returned by \Call{Distributed-Spanner}{} has a weight of $\mathcal{O}(1)\weight(MST)$.\label{prp:dis-light}
\end{proposition}
\begin{proof}
By Corollary \ref{cor:dis-cmp1}, Lemma \ref{lem:dis-cmp2}, and Lemma \ref{lem:dis-cmp3},
$$\weight(\mathcal{S}_w) = \mathcal{O}(1)\weight(\mathcal{S}^*_w)$$
Summing up together these inequalities for $w\in I$,
$$\weight(\text{output}) = \mathcal{O}(1)\sum_{w\in I}\weight(\mathcal{S}^*_w)$$
But we recall that every vertex, and hence every edge of $\mathcal{S}^*$, is repeated $\mathcal{O}(1)$ times in the summation above, so
$$\weight(\text{output}) = \mathcal{O}(1)\weight(\mathcal{S}^*) = \mathcal{O}(1)\weight(MST(G))$$
\end{proof}

Therefore we have all the ingredients to prove \autoref{thm:dis}.
\begin{theorem}[Distributed Spanner]
\thmDistributed\label{thm:dis}
\end{theorem}
\begin{proof}
It directly follows from Lemma \ref{lem:dis-rnd}, Lemma \ref{lem:dis-str}, Lemma \ref{lem:dis-deg}, and Proposition \ref{prp:dis-light}.
\end{proof}


\section{Adjustments for the CONGEST Model} \label{sec:congest}
In this section we study the problem of finding a bounded-degree $(1+\epsilon)$-spanner in the CONGEST model of computation, for a point set that is located in a doubling metric space. In the CONGEST model, every node can send a message of bounded size to every other node in a single round of communication. This makes it hard to gather any global information about the graph.

The maximal independent set algorithm of \cite{kuhn2005locality} still works in $\mathcal{O}(\log^*n)$ rounds of communication in the CONGEST model. But our proposed distributed algorithm (Algorithm \ref{alg:localgreedy}) needs to gather 2-hop neighborhood information of every center in the MIS, which requires $\mathcal{O}(D)$ rounds in the CONGEST model, where $D$ is the maximum degree of a vertex in the unit ball graph. The rest of the algorithm is performed locally and the number edges sent to every neighbor in the end is bounded by a constant, so the remaining of the algorithm only requires a constant number of rounds.

It is natural to ask whether our algorithm can be adapted to the CONGEST model, and if it requires more communication rounds compared to the LOCAL model. In this section we show how to modify our algorithm to work in the CONGEST model, and surprisingly, have no asymptotic change on its number of communication rounds.

As we mentioned earlier, the only step of our algorithm that requires more than constant rounds of communication is the aggregation of the 2-hop neighborhood information for every center in the MIS. We passed the 2-hop neighborhoods to our centralized algorithm to find an asymptotically optimal spanner on them, which was later distributed among the vertices in the neighborhood to form the final spanner. Here, for our construction in the CONGEST model, we directly address the problem of finding an asymptotically optimal spanner for the 2-hop neighborhoods, without the need to access all of the points in those neighborhoods.

Let $w\in I$ be a center in the maximal independent set. We partition the edges of the UBG on $\mathcal{N}^2(w)$ into two sets, depending on whether their length is larger than $1/2$ or not. We aim to find asymptotically optimal spanners for each partition separately. We use the notation $G_{\leq\alpha}$ to refer to the subgraph of the unit ball graph that consists of edges of length at most $\alpha$. We similarly define $G_{>\alpha}$. Therefore, we can refer to the subgraphs induced by the two partitions by $G_{\leq1/2}$ and $G_{>1/2}$.

First, we show that in constant rounds of communication, we can find a covering of the points in $\mathcal{N}^2(w)$ with at most a constant number of balls of radius $1/2$. The existence of such covering trivially follows from the definition of a doubling metric space, but finding such covering in the distributed setting is not trivial. Therefore, we introduce the following procedure: Every center $v\in\mathcal{N}^1(w)$ (including $w$ itself) finds a maximal independent set $I_{1/4}(v)$ of the vertices $\mathcal{N}^1(v)$ in $G_{\leq1/4}$, and sends it to $w$, all centers at the same time. Recall that $\mathcal{N}^1(v)$ is the set of neighbors of $v$ in the UBG, and a maximal independent set in $G_{\leq1/4}$ is simply a maximal set of vertices where the pair-wise distance of each two vertex is at least $1/4$. Finding this maximal independent set for each $v$ can be easily done using a (centralized) greedy algorithm, and the size of such maximal independent set would be bounded by a constant according to the packing lemma. Therefore, this step can be done in constant number of rounds. Afterwards, $w$ calculates a maximal independent set $\mathcal{I}(w)$ of the vertices $\cup_{v\in\mathcal{N}^1(w)}I_{1/4}(v)$ in $G_{\leq1/4}$. We show that the centers in $\mathcal{I}$ satisfy our desired properties.

\def\lemCovering{The union of the balls of radius $1/2$ around the centers in $\mathcal{I}(w)$ cover $\mathcal{N}^2(w)$. Furthermore, the size of $\mathcal{I}(w)$ is bounded by a constant.}%
\begin{lemma}%
\label{lem:covering}%
\lemCovering
\end{lemma}
\begin{proof}
Let $v\in\mathcal{N}^2(w)$ be an arbitrary point. Thus there exists $u\in\mathcal{N}^1(w)$ that $v\in\mathcal{N}^1(u)$. Let $I_{1/4}(u)$ be the maximal independent set of the vertices $\mathcal{N}^1(u)$ in $G_{\leq1/4}$, that $u$ calculates and sends to $w$ in the first step. There exists $v'\in I_{1/4}(u)$ that $\lVert vv'\rVert\leq 1/4$. Similarly, there exists $v''\in\mathcal{I}(w)$ that $\lVert v'v''\rVert\leq 1/4$. By the triangle inequality, $\lVert vv''\rVert\leq 1/2$, i.e. $v$ is covered by a ball of radius 1/2 around $v''$.

On the other hand, $\mathcal{I}(w)$ is contained in a ball of radius 2 and every pair of points in $\mathcal{I}(w)$ have a distance of at least $1/4$. Thus, by the packing lemma, he size of $\mathcal{I}(w)$ is bounded by a constant.
\end{proof}

Next, every center $v\in \mathcal{I}(w)$ calculates a $(1+\epsilon)$-spanner $\mathcal{S}_{\leq1/2}(v)$ of the point set $\mathcal{N}^1(v)$ using the centralized algorithm, and notifies its neighbors about their connections. We prove that the union of these spanners, would be a spanner for one of the partitions, i.e. the edges of length at most $1/2$ in $\mathcal{N}^2(w)$. The pseudo-code of this procedure is available in \autoref{alg:shortedges}

\def\lemCongestSpannerOne{The union of the spanners $\mathcal{S}_{\leq1/2}(v)$ for $v\in\mathcal{I}(w)$ is a $(1+\epsilon)$-spanner of $\mathcal{N}^2(w)$ in $G_{\leq1/2}$. The maximum degree and the lightness of this spanner are both bounded by constants.}%
\begin{lemma}%
\label{lem:congest-spanner1}%
\lemCongestSpannerOne
\end{lemma}
\begin{proof}
First, we prove the $1+\epsilon$ stretch-factor. Let $(u,v)$ be a pair in $\mathcal{N}^2(w)$ such that $\lVert uv\rVert\leq 1/2$. By Lemma \ref{lem:covering} we know there exists $u'\in\mathcal{I}(w)$ that $\lVert uu'\rVert\leq1/2$. Thus $\lVert vu'\rVert\leq 1$ which means that $u,v\in\mathcal{N}^1(u')$ and there would be a $(1+\epsilon)$-path for this pair in $\mathcal{S}_{\leq1/2}(u')$, which would be present in the union of the spanners.

The degree bound follows from the fact that, by the packing lemma, every point in $\mathcal{N}^2(w)$ is appeared in at most a constant number of one-hop neighborhoods and therefore in at most a constant number of spanners constructed the elements in $\mathcal{I}(w)$. Since in every spanner it has a bounded degree, in the union it will have a bounded degree as well.

To prove the lightness bound we follow a similar approach to the proof of Proposition \ref{prop:lightness}. The weight of each spanner $\mathcal{S}_{\leq1/2}(v)$ is $\mathcal{O}(1)\weight(MST(\mathcal{N}^1(v)))$. The weight of the MST is at most twice the weight of the optimal Steiner tree on $\mathcal{N}^2(v)$ with the required vertices being $\mathcal{N}^1(v)$. And the weight of this optimal Steiner tree is at most equal to the weight of the induced sub-graph of an (asymptotically) optimal $(1+\epsilon)$-spanner of $G$ on the subset of vertices $\mathcal{N}^2(v)$. Summing up these subgraphs for different $v$s and different $w$s would end up with adding at most a constant factor to the weight of the optimal spanner, which proves that the lightness would be bounded by a constant.
\end{proof}

\begin{algorithm}[ht]
\caption{The CONGEST spanner algorithm.}\label{alg:congest}
\textbf{Input.} A unit ball graph $G(V, E)$ in a metric with doubling dimension $d$ and an $\epsilon>0$.\\
\textbf{Output.} A light-weight bounded-degree $(1+\epsilon)$-spanner of $G$.%
\begin{algorithmic}[1]%
\Procedure{CONGEST-Spanner}{$G$, $\epsilon$}%
\State Find a maximal independent set $I$ of $G$ using \cite{kuhn2005locality}%
\State Run \Call{Span-Short-Edges}{} on the vertices of $G$%
\State Run \Call{Span-Long-Edges}{} on the vertices of $G$%
\EndProcedure%
\end{algorithmic}
\end{algorithm}

\begin{algorithm}[ht]
\caption{Finding a spanner of the edges of length smaller than $1/2$ in $\mathcal{N}^2(w)$.}\label{alg:shortedges}
\begin{algorithmic}[1]%
\Function{Span-Short-Edges}{vertex $u$}%
\If {$u\in I$}%
\State Send a signal of type 1 to every $v\in\mathcal{N}^1(u)$.%
\State Wait for their maximal independent sets, $I_{1/4}(v)$s.%
\State Calculate a maximal independent set of $\cup_{v\in\mathcal{N}^1(u)}I_{1/4}(v)$ in $G_{<1/4}$ greedily.%
\State Store this maximal independent set in $\mathcal{I}(u)$.%
\State Send a signal of type 2 to every $v\in\mathcal{I}(u)$.%
\EndIf%
\If {received type 1 signal from some $w$}%
\State Calculate a maximal independent set of $\mathcal{N}^1(w)$ in $G_{1/4}$ greedily.%
\State Send this maximal independent set to $w$.%
\EndIf%
\If {received type 2 signal from some $w$}%
\State Calculate $\mathcal{S}_{\leq1/2}(u)\gets$ \Call{Centralized-Spanner}{$\mathcal{N}^1(u)$, $\epsilon$}%
\For{$e=(a,b)$ in $\mathcal{S}_{\leq1/2}(u)$}%
\State Send $e$ to $a$ and $b$%
\EndFor%
\EndIf%
\State Receive and store the edges sent by other centers%
\EndFunction%
\end{algorithmic}
\end{algorithm}

Now we find a spanner for the other partition, the edges of length larger than $1/2$ in $\mathcal{N}^2(w)$. The procedure is as follows: First, every center $v\in\mathcal{N}^1(w)$ calculates a maximal independent set $I_{\epsilon/40}(v)$ of $\mathcal{N}^1(v)$ in $G_{\leq\epsilon/40}$ and sends it to $w$. Again, the size of each maximal independent set is $\mathcal{O}(\epsilon^{-d})$ by the packing lemma, which is constant. Therefore, this step takes only constant number of rounds. Afterwards, $w$ finds a maximal independent set $\mathcal{I}'(w)$ of $\cup_{v\in\mathcal{N}^1(w)}I_{\epsilon/40}(v)$ in $G_{\leq\epsilon/40}$. Then $w$ constructs a $(1+\epsilon/5)$-spanner $\mathcal{S}'(w)$ of $\mathcal{I}'(w)$ in $G$ using the centralized algorithm, and announces the edges of the spanner to their corresponding endpoints. Finally, every center $v\in\mathcal{I}'(w)$ calculates a $(1+\epsilon)$-spanner $\mathcal{S}''(v)$ of its $\epsilon/20$ neighborhood and announces its edges to their endpoints. We show that the union of $S'(w)$ and $S''(v)$s for $v\in\mathcal{I}'(w)$ would form a $(1+\epsilon)$-spanner of the second partition, i.e. the edges of larger than $1/2$. The pseudo-code of this procedure is available in \autoref{alg:longedges}

\begin{lemma}%
\label{lem:coveringsmall}%
The union of the balls of radius $\epsilon/20$ around the centers in $\mathcal{I}'(w)$ cover $\mathcal{N}^2(w)$. Furthermore, the size of $\mathcal{I}'(w)$ is bounded by a constant.
\end{lemma}
\begin{proof}
Similar to the proof of Lemma \ref{lem:covering}.
\end{proof}

\def\lemCongestSpannerTwo{The union of the spanners $\mathcal{S}'(w)$ and $\mathcal{S}''(v)$ for $v\in\mathcal{I}'(w)$ forms a $(1+\epsilon)$-spanner of $\mathcal{N}^2(w)$ in $G_{>1/2}$. The maximum degree and the lightness of this spanner are btoh bounded by constants.}%
\begin{lemma}%
\label{lem:congest-spanner2}%
\lemCongestSpannerTwo
\end{lemma}
\begin{proof}
Again, we first prove the $1+\epsilon$ stretch-factor of the spanner. Let $(u,v)$ be a pair in $\mathcal{N}^2(w)$ that $\lVert uv\rVert> 1/2$. Let $u'$ and $v'$ be centers in $\mathcal{I}'(w)$ that are at distance of at most $\epsilon/20$ from $u$ and $v$, respectively. Such centers exist according to Lemma \ref{lem:coveringsmall}. Consider the $(1+\epsilon)$-path connecting $u$ to $u'$ in $\mathcal{S}''(u')$ and the $(1+\epsilon)$-path connecting $v$ to $v'$ in $\mathcal{S}''(v')$. We can attach these paths together with the $(1+\epsilon/5)$-path between $u'$ and $v'$ in $\mathcal{S}'(w)$ to get a path between $u$ and $v$. The stretch of this path would be at most
$$\frac{(1+\epsilon)(\lVert uu'\rVert + \lVert vv'\rVert) + (1+\epsilon/5)\lVert u'v'\rVert}{\lVert uv\rVert} = \frac{(1+\epsilon)(\lVert uu'\rVert + \lVert vv'\rVert)}{\lVert uv\rVert} + \frac{(1+\epsilon/5)\lVert u'v'\rVert}{\lVert uv\rVert}$$
But,
$$\frac{(1+\epsilon)(\lVert uu'\rVert + \lVert vv'\rVert)}{\lVert uv\rVert}\leq \frac{(1+\epsilon)\epsilon/10}{1/2}\leq \frac{2\epsilon}{5}$$
Also,
$$\frac{(1+\epsilon/5)\lVert u'v'\rVert}{\lVert uv\rVert} \leq \frac{(1+\epsilon/5)(\lVert uv\rVert + \lVert uu'\rVert + \lVert vv'\rVert)}{\lVert uv\rVert}\leq 1+\epsilon/5+\frac{(1+\epsilon/5)\epsilon/10}{\lVert uv\rVert}$$
Bounding the last term,
$$\frac{(1+\epsilon/5)\epsilon/10}{\lVert uv\rVert}\leq \frac{(6/5)\epsilon/10}{1/2}=\frac{6\epsilon}{25}$$
Therefore, the stretch of the $uv$-path would be upper bounded by,
$$\frac{2\epsilon}{5} + 1+\epsilon/5 + \frac{6\epsilon}{25} < 1+\epsilon$$

An approach similar to the proof of Lemma \ref{lem:congest-spanner1} shows that the degree of every vertex in the union of $\mathcal{S}''(v)$s would be bounded by a constant. We do not repeat the details of the proof here. From the properties of our centralized construction, the degree of every vertex would be bounded in $\mathcal{S}'(w)$ as well. Thus, the degree of every vertex in the union of these spanners would be bounded by a constant.

To prove the lightness bound, we bound the weight of each spanner separately. First, we bound the total weight of $\mathcal{S}'(w)$. We know from the properties of our centralized construction, that $\weight(\mathcal{S}'(w))=\mathcal{O}(1)\weight(MST(\mathcal{I}'(w)))$. But $\weight(MST(\mathcal{I}'(w)))\leq 2MST(\mathcal{N}^2(w))$, so $\weight(\mathcal{S}'(w))=\mathcal{O}(1)\weight(MST(\mathcal{N}^2(w)))$. Therefore, by Lemma \ref{lem:dis-cmp2} and Lemma \ref{lem:dis-cmp3} the total weight of $\mathcal{S}'(w)$ spanners for different centers $w$ would sum up to at most a constant factor of the weight of the optimal spanner.

Next, we bound the total weight of $\mathcal{S}''(v)$ spanners. Again, we know from the properties of our centralized construction that $\weight(\mathcal{S}''(v))=\mathcal{O}(1)\weight(MST(\mathcal{N}^{\epsilon/20}(v)))$. Assuming that $\mathcal{S}^*$ is an optimal spanner on the point set, we can observe that any $(1+\epsilon)$-path (in $\mathcal{S}^*$) between any pair of vertices in $\mathcal{N}^{\epsilon/20}(v)$ must be completely contained in a ball of radius $3\epsilon/20$, otherwise the length of the path would be more than $(1+\epsilon)\epsilon/10$, the maximum allowed length for any $(1+\epsilon)$-path of any pair in the $\mathcal{N}^{\epsilon/20}(v)$ neighborhood. Therefore, the induced sub-graph of $\mathcal{S}^*$ on $\mathcal{N}^{3\epsilon/20}(v)$ has a connected component connecting the vertices of $\mathcal{N}^{\epsilon/20}(v)$. Thus, its weight is at least equal to the weight of a minimum Steiner tree on $\mathcal{N}^{3\epsilon/20}(v)$, with the required vertices being $\mathcal{N}^{\epsilon/20}(v)$. This is at least equal to $\weight(MST(\mathcal{N}^{\epsilon/20}(v)))/2$. Therefore, the weight of $\mathcal{S}''(v)$ is bounded above by a constant factor of the weight of the induced sub-graph of $\mathcal{S}^*$ on $\mathcal{N}^{3\epsilon/20}(v)$. Summing up these bounds for every $v$ in every $w$ would lead to at most a constant repetitions of every vertex and every edge (similar to Proposition \ref{prop:lightness}) in $\mathcal{S}^*$, which shows that the total weight of $\mathcal{S}''(v)$ for different vertices of $v$ would be bounded by a constant factor of the weight of the optimal spanner.
\end{proof}

\begin{algorithm}[ht]
\caption{Finding a spanner of the edges of length larger than $1/2$ in $\mathcal{N}^2(w)$.}\label{alg:longedges}
\begin{algorithmic}[1]%
\Function{Span-Long-Edges}{vertex $u$}%
\If {$u\in I$}%
\State Send a signal of type 3 to every $v\in\mathcal{N}^1(u)$.%
\State Wait for their maximal independent sets, $I_{\epsilon/40}(v)$s.%
\State Calculate a maximal independent set of $\cup_{v\in\mathcal{N}^1(u)}I_{\epsilon/40}(v)$ in $G_{<\epsilon/40}$ greedily.%
\State Store this maximal independent set in $\mathcal{I}'(u)$.%
\State Send a signal of type 4 to every $v\in\mathcal{I}'(u)$.%
\State Calculate $\mathcal{S}'(u)\gets$ \Call{Centralized-Spanner}{$\mathcal{I}'(u)$, $\epsilon/5$}%
\For{$e=(a,b)$ in $\mathcal{S}'(u)$}%
\State Send $e$ to $a$ and $b$%
\EndFor%
\EndIf%
\If {received type 3 signal from some $w$}%
\State Calculate a maximal independent set of $\mathcal{N}^1(w)$ in $G_{\epsilon/40}$ greedily.%
\State Send this maximal independent set to $w$.%
\EndIf%
\If {received type 4 signal from some $w$}%
\State Let $\mathcal{N}^{\epsilon/20}(u)$ be the $\epsilon/20$ neighborhood of $u$, i.e. the set of vertices that are at distance $\epsilon/20$ or less from $u$.%
\State Calculate $\mathcal{S}''(u)\gets$ \Call{Centralized-Spanner}{$\mathcal{N}^{\epsilon/20}(u)$, $\epsilon$}%
\For{$e=(a,b)$ in $\mathcal{S}''(u)$}%
\State Send $e$ to $a$ and $b$%
\EndFor%
\EndIf%
\State Receive and store the edges sent by other centers%
\EndFunction%
\end{algorithmic}
\end{algorithm}

The union of the two spanners for the two partitions form a spanner for the 2-hop neighborhood of $w$, the goal we wanted to achieve in the CONGEST model. This completes our adjustments in this model.

\begin{theorem}[CONGEST Spanner]
\thmCONGEST\label{thm:congest}
\end{theorem}
\begin{proof}
The proof follows from Lemma \ref{lem:congest-spanner1} and Lemma \ref{lem:congest-spanner2}.
\end{proof}


\section{Low-Intersection Construction} \label{sec:euc}
The two dimensional Euclidean case of the unit ball graphs, also known as unit disk graphs, had gained a significant attention once it was introduced, due to its direct application in wireless ad-hoc networks. A huge amount of effort is still being made to improve the existing spanners in some aspects, or to introduce new constructions that possess good qualities, e.g. being fault tolerant, or having a low interference \cite{mulzer2020compact,wu2018topology,chan2019approximate,yu2017distributed}.

In this section, we take the edge intersection as a simple representation of physical link-to-link interference, and we provide a distributed construction for a light-weight low-intersection bounded-degree spanner for unit disk graphs in the two dimensional Euclidean plane. To the best of our knowledge, this is the first distributed low-intersection construction with constant bounds on the degree and lightness.

We need to emphasize that after removing the edges of longer than 1 from the output of \Call{Naive-Greedy}{} on a unit disk graph $G$, the remaining graph would be a $(1+\epsilon)$-spanner of $G$ and it also has constant bounds on its lightness, maximum degree, and average number of edge intersections per node. Equivalently, we can stop the algorithm before reaching the pairs with distance greater than 1 and the resulting spanner would be the same. We present this centralized algorithm in the following form and we use it as a part of our distributed algorithm.

\begin{algorithm}[ht]
\caption{The centralized construction for a low-intersection spanner of UDG.}\label{alg:euc-central}
\textbf{Input.} A unit disk graph $G(V,E)$ in the two dimensional Euclidean plane.\\
\textbf{Output.} A light-weight low-intersection bounded-degree $(1+\epsilon)$-spanner of $G$.%
\begin{algorithmic}[1]%
\Procedure{Centralized-Euclidean-Spanner}{$G$, $\epsilon$}%
\State Let $S$ be a graph with vertices $V$ and edges $E=\{\}$%
\For {each $(P,Q)\in V^2$ in increasing order of $\dist(P,Q)$}%
\If {$\dist(P,Q) > 1$}%
\State \textbf{break}%
\EndIf%
\If {$\dist_S(P,Q) > t\cdot \dist(P,Q)$}
\State Add edge $PQ$ to $E$%
\EndIf%
\EndFor%
\Return S%
\EndProcedure%
\end{algorithmic}
\end{algorithm}

\begin{theorem}[Centralized Euclidean Spanner]
\thmEuclideanCentralized\label{thm:euc-cent}
\end{theorem}
\begin{proof}
The stretch-factor follows directly from the fact that every edge of the UDG can be approximated by a spanner path within a factor of $t$ (otherwise the edge would have been added to the spanner). And since the spanner is a sub-graph of the greedy spanner on the complete weighted graph on $V$, we can deduce that its lightness, its maximum degree, and its average number of edge intersections per node are bounded by the lightness, maximum degree, and the average number of edge intersections per node of the greedy spanner, respectively. And these are all bounded by constants according to \cite{filtser2016greedy} and \cite{eppstein2020edge}.
\end{proof}

\subsection{The algorithm}

The main result of this section is the distributed construction. Our distributed algorithm (\autoref{alg:euc-dist}) is based on the algorithm we presented in \autoref{sec:dist}, with a small change that we use the \Call{Centralized-Euclidean-Spanner}{} procedure as a sub-routine instead of \Call{Centralized-Spanner}{}. The pseudo-code is provided in \autoref{alg:euc-dist}.

\begin{algorithm}[ht]
\caption{The distributed construction for a low-intersection spanner of UDG.}\label{alg:euc-dist}
\textbf{Input.} A unit disk graph $G(V,E)$ in the two dimensional Euclidean plane.\\
\textbf{Output.} A light-weight bounded-degree $(1+\epsilon)$-spanner of $G$.%
\begin{algorithmic}[1]%
\Procedure{Distributed-Euclidean-Spanner}{$G$, $\epsilon$}%
\State Find a maximal independent set $I$ of $G$ using \cite{kuhn2005locality}%
\State Run \Call{Local-Greedy}{} on the vertices of $G$%
\EndProcedure%
\Function{Local-Greedy}{vertex $w$}%
\State Retrieve $\mathcal{N}^2(w)$, the 2-hop neighborhood information of $w$%
\If {$w$ is in $I$}%
\State $\mathcal{S}_w \gets $ \Call{Centralized-Euclidean-Spanner}{$\mathcal{N}^2(w)$, $\epsilon$}%
\For{$e=(u,v)$ in $\mathcal{S}_w$}%
\State Send $e$ to $u$ and $v$%
\EndFor%
\EndIf%
\State Listen to incoming edges and store them%
\EndFunction%
\end{algorithmic}
\end{algorithm}

Similar to \autoref{alg:localgreedy} this algorithm runs in $\mathcal{O}(\log^* n)$ rounds of communications during which it builds a $(1+\epsilon)$-spanner of the unit disk graph $G$.
\begin{lemma}
The spanner returned by \Call{Distributed-Euclidean-Spanner}{} has a stretch-factor of $1+\epsilon$, a weight of $\mathcal{O}(1)\weight(MST)$, and a maximum degree of $\mathcal{O}(1)$.\label{lem:euc-dist-prop}
\end{lemma}
\begin{proof}
The proof follows from Lemma \ref{lem:dis-str}, Proposition \ref{prp:dis-light}, and Lemma \ref{lem:dis-deg} which all hold for this construction as well.
\end{proof}

Now we prove the low-intersection property of our distributed construction.
\begin{proposition}
The spanner returned by \Call{Distributed-Euclidean-Spanner}{} has at most a linear number of edge intersections.\label{prp:euc-dist-int}
\end{proposition}
\begin{proof}
We generalize the result of Eppstein and Khodabandeh \cite{eppstein2020edge} on the edge crossings of the greedy spanner, which states that an arbitrary edge $AB$ of the greedy spanner intersects with at most a constant number of longer edges. The main observation is that the same statement is true when $AB$ is an arbitrary segment on the plane, and the assumption of $AB$ being an edge of the spanner could be eliminated from the proof of the theorem. The modified statement of the theorem for our case would be as follows,

\begin{lemma}
Given an arbitrary segment $\lVert AB\rVert \leq 1$ in the plane, the number of edges $e\in S$ that intersect $AB$ and $|e|\geq \lVert AB\rVert$ is bounded by a constant.\label{lem:int-ES}
\end{lemma}

We defer the geometric proof of this lemma to a separate section in Appendix \ref{sec:geometric}. Here we use this lemma to bound the number of intersections of each edge of the spanner with the longer edges, which in turn proves the linear bound on the number of edge intersections in total.

Let $e=(u,v)$ be an edge of the final spanner and let $f=(u',v')\in S_w$ for some $w$ be another edge that intersects $e$. By the triangle inequality, at least one of $u'$ and $v'$ needs to be adjacent to $u$ in $G$. Thus $u\in\mathcal{N}^3(w)$, which means that $\lVert uw\rVert \leq 3$. So by the packing property there are at most $12^2=144$ different choices for $w$ that $u\in\mathcal{N}^3(w)$. But for every such $w$, according to Lemma \ref{lem:int-ES} for the segment $AB=e$, there are at most a constant number of intersections between $e$ and the edges of length larger than $|e|$ in $S_w$. Since there are at most a constant number of choices for $w$ and for each choice there are at most a constant number of intersections between $e$ and the longer edges in $S_w$, we conclude that the number of intersections between $e$ and longer edges in the final spanner would be bounded by a constant, which indeed proves the proposition.
\end{proof}

Thus the following theorem holds for our distributed spanner construction for the two dimensional Euclidean plane.
\begin{theorem}[Distributed Euclidean Spanner]
\thmEuclideanDistributed\label{thm:euc-dist}
\end{theorem}
\begin{proof}
Follows directly from Lemma \ref{lem:euc-dist-prop} and Proposition \ref{prp:euc-dist-int}.
\end{proof}

This low-intersection property also implies the existence of small separators for our spanner, which is stated in the following corollary.
\begin{corollary}
The spanner returned by \Call{Distributed-Euclidean-Spanner}{} has a separator of size $\mathcal{O}(\sqrt{n})$, where $n$ is the number of vertices. Also, a separator hierarchy can be constructed from its planarization in linear time.
\end{corollary}
\begin{proof}
Lemma \ref{lem:int-ES} implies that the crossing graph of our spanner has a bounded degeneracy. Therefore, this corollary follows from the result of Eppstein and Gupta \cite{eppstein2017crossing}.
\end{proof}

\subsection{Higher dimensions}

In higher dimensions of Euclidean spaces, it does not particularly make sense to talk about the edge intersections of the spanner, as the edges would not intersect for points in general locations. But we can generalize our separator result to higher dimensions of Euclidean spaces.

Recently, Li and Than \cite{le2022greedy} proved that any geometric graph with a property that they called $\tau$-lanky has separators of size $\mathcal{O}(\tau n^{1-1/d})$, where $d$ is the dimension of the space. They also proved that greedy spanners are $\mathcal{O}(1)$-lanky and therefore their $k$-vertex subgraphs have separators of size $\mathcal{O}(k^{1-1/d})$. We can take advantage of this result for higher dimensions and prove the existence of small separators for our construction in higher dimensions of Euclidean spaces.

\begin{lemma}%
\label{lem:lanky}%
Let $S$ be the output of $\Call{Distributed-Euclidean-Spanner}{}$ on a set of points in the $d$-dimensional Euclidean space. Given an arbitrary ball of radius $R\leq 1$ in this space, the number of edges $e\in S$ that cut this ball would be bounded by a constant.
\end{lemma}
\begin{proof}
A similar proof to the proof of lemma \ref{lem:int-ES} yields this result in higher dimensions as well. This result is also proven for the centralized greedy algorithm in \cite{le2022greedy}, but the extension to the distributed setting needs some considerations, which are similar to the proof of Lemma \ref{lem:int-ES}, and not included to avoid repetition.
\end{proof}

According to lemma \ref{lem:lanky} our construction in $d$-dimensional Euclidean space is $\mathcal{O}(1)$-lanky and therefore, it possesses separators of small size.
\begin{corollary}
The spanner returned by $\Call{Distributed-Euclidean-Spanner}{}$ in the $d$-dimensional Euclidean space is $\mathcal{O}(1)$-lanky. Therefore, any $k$-vertex subgraph of this construction possesses separators of size $\mathcal{O}(k^{1-1/d})$. Also, a separator hierarchy can be built for this spanner in expected linear time.
\end{corollary}
\begin{proof}
The $\mathcal{O}(1)$-lanky property follows from lemma \ref{lem:lanky}, and the existence of separators of size $\mathcal{O}(k^{1-1/d})$ and the expected time on finding a separator hierarchy follows from the result of \cite{le2022greedy}.
\end{proof}

\section{Experimental Results}
\label{sec:exp}
In this section we provide empirical evidence for the efficiency of our distributed construction. While we have proven rigorous bounds on the lightness and sparsity of our spanner, it might be unclear how it performs in practice compared to an efficient centralized construction. We design an experiment in the two dimensional Euclidean plane that answers this question.

We run our distributed spanner algorithm (Algorithm \ref{alg:euc-dist}) on a point set consisting of 100 points uniformly chosen at random from a 5 by 5 square in the two dimensional Euclidean plane. This point set, together with its unit disk graph, is drawn in Figure \ref{fig:exp-1}. In the first part of the experiment, we run our distributed algorithm for different values of $t$, the stretch-factor, and we compare the result of our algorithm with the output of the centralized greedy spanner on the same point set, with the same parameter $t$.

\begin{figure}[ht]
    \centering
    \includegraphics[trim={1.7cm 1.7cm 1.5cm 1.7cm},clip,width=0.3\textwidth]{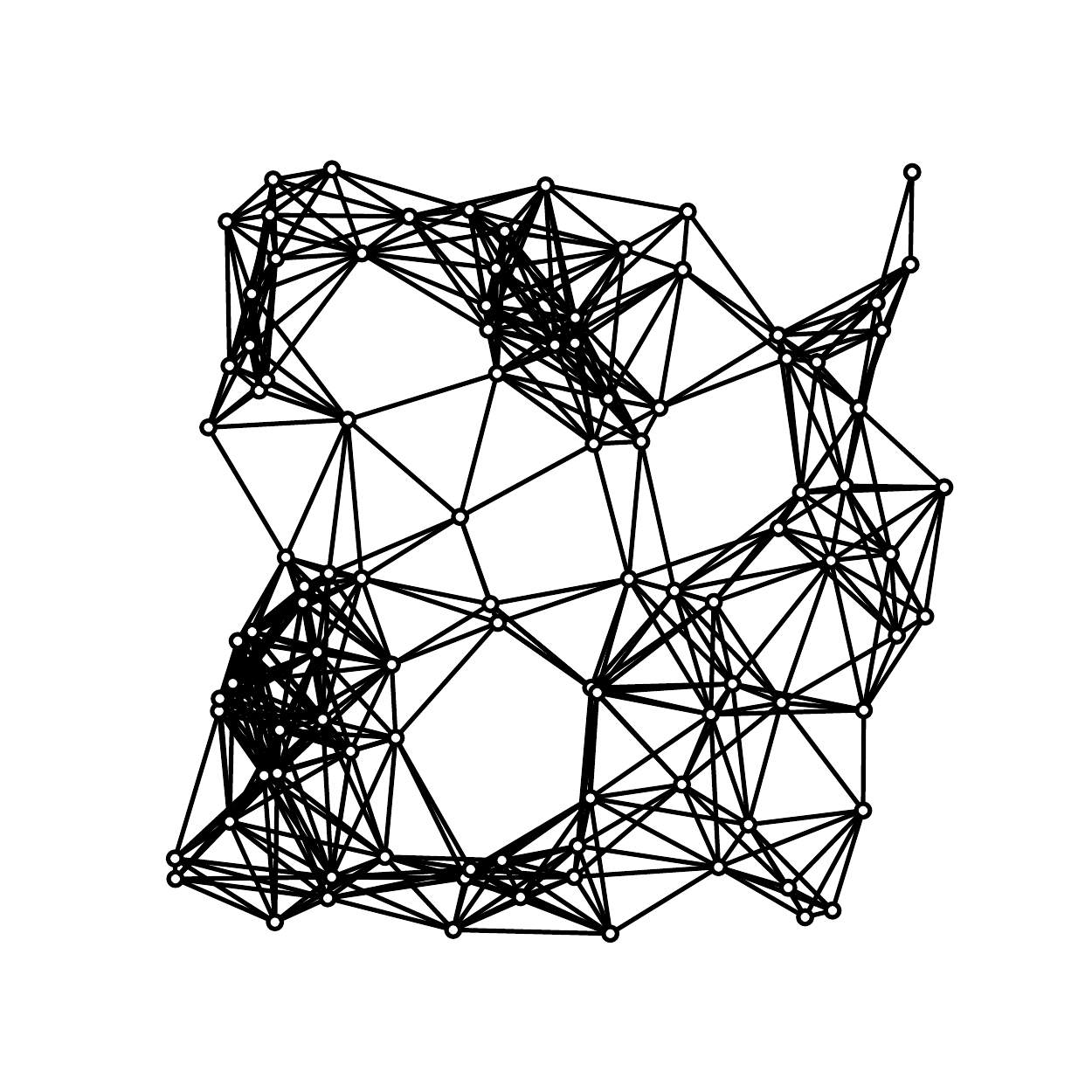}
    \caption{The random point set and its unit disk graph, from the first part of our experiment.}
    \label{fig:exp-1}
\end{figure}

The results of this (Figure \ref{fig:exp-2}) shows a near-optimal performance from our distributed algorithm. As we mentioned earlier, the centralized greedy spanner is known for its near-optimality, and our construction is comparable with this near-optimal solution.

\begin{figure}[ht]
    \centering
    \begin{subfigure}[b]{0.4\textwidth}
        \centering
        \includegraphics[trim={0.2cm 0cm 1cm 0.7cm},clip,width=\textwidth]{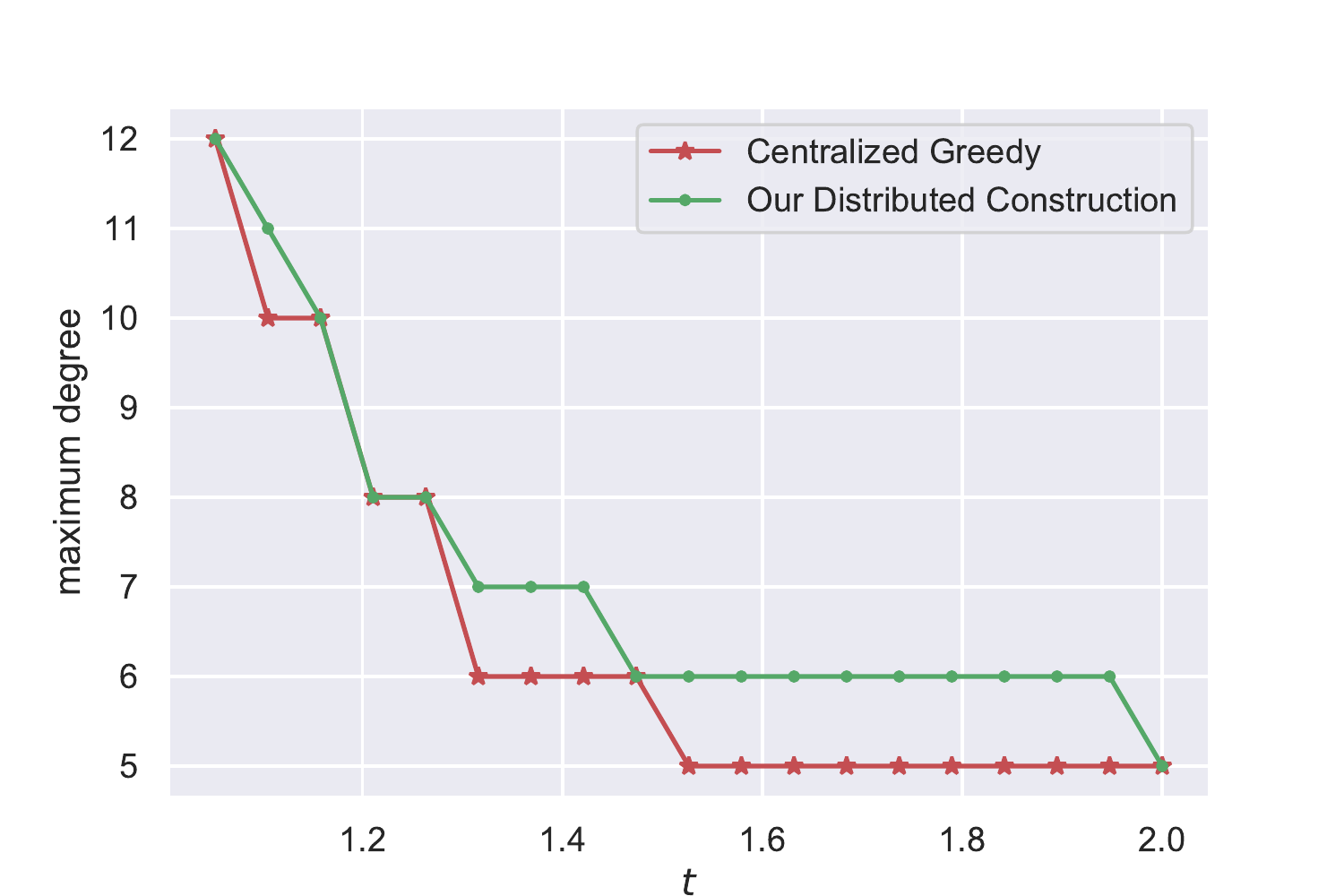}
        \caption{Degree comparison}
    \end{subfigure}
    \begin{subfigure}[b]{0.4\textwidth}
        \centering
        \includegraphics[trim={0.2cm 0cm 1cm 0.7cm},clip,width=\textwidth]{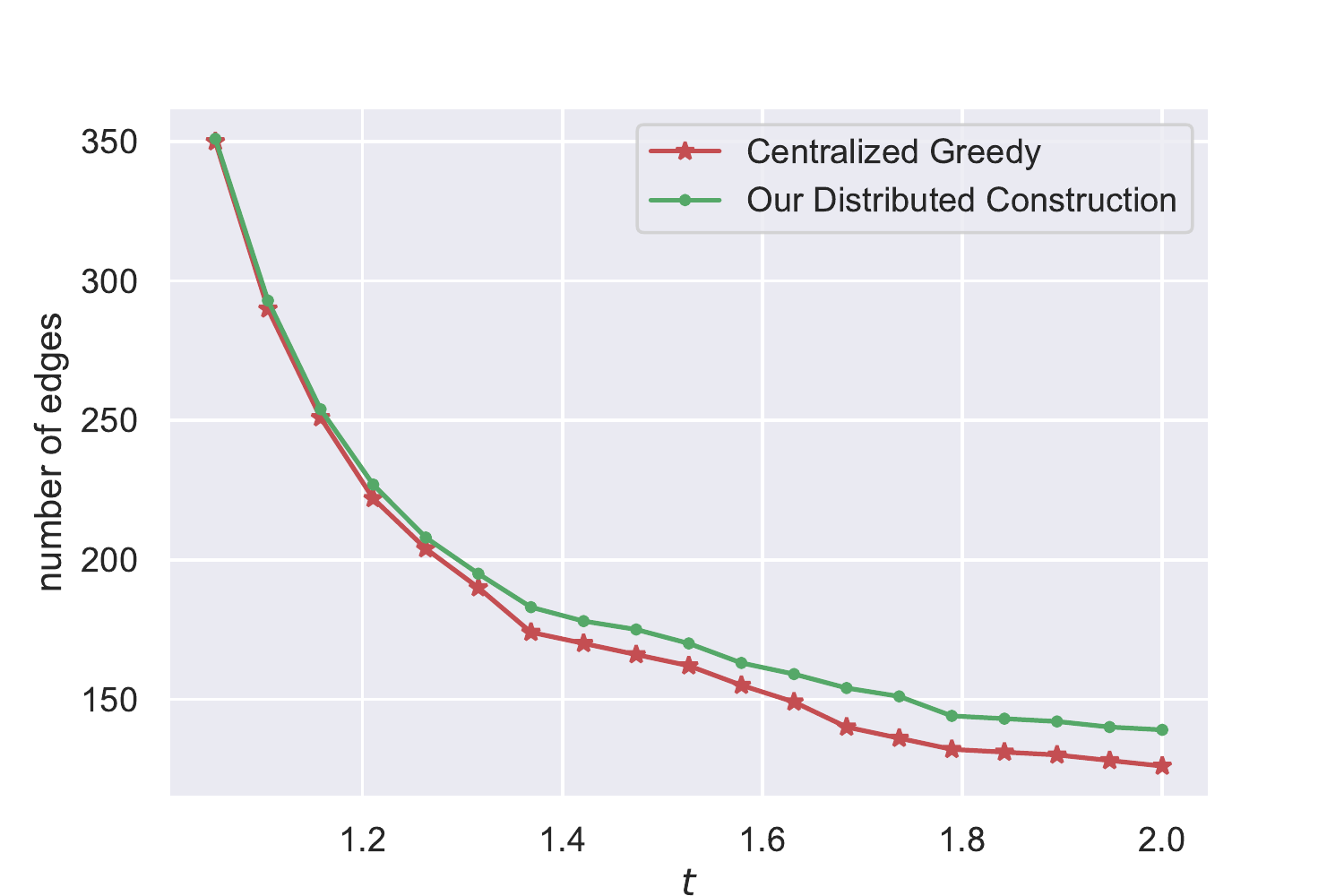}
        \caption{Size comparison}
    \end{subfigure}
    \begin{subfigure}[b]{0.4\textwidth}
        \centering
        \includegraphics[trim={0.2cm 0cm 1cm 0.7cm},clip,width=\textwidth]{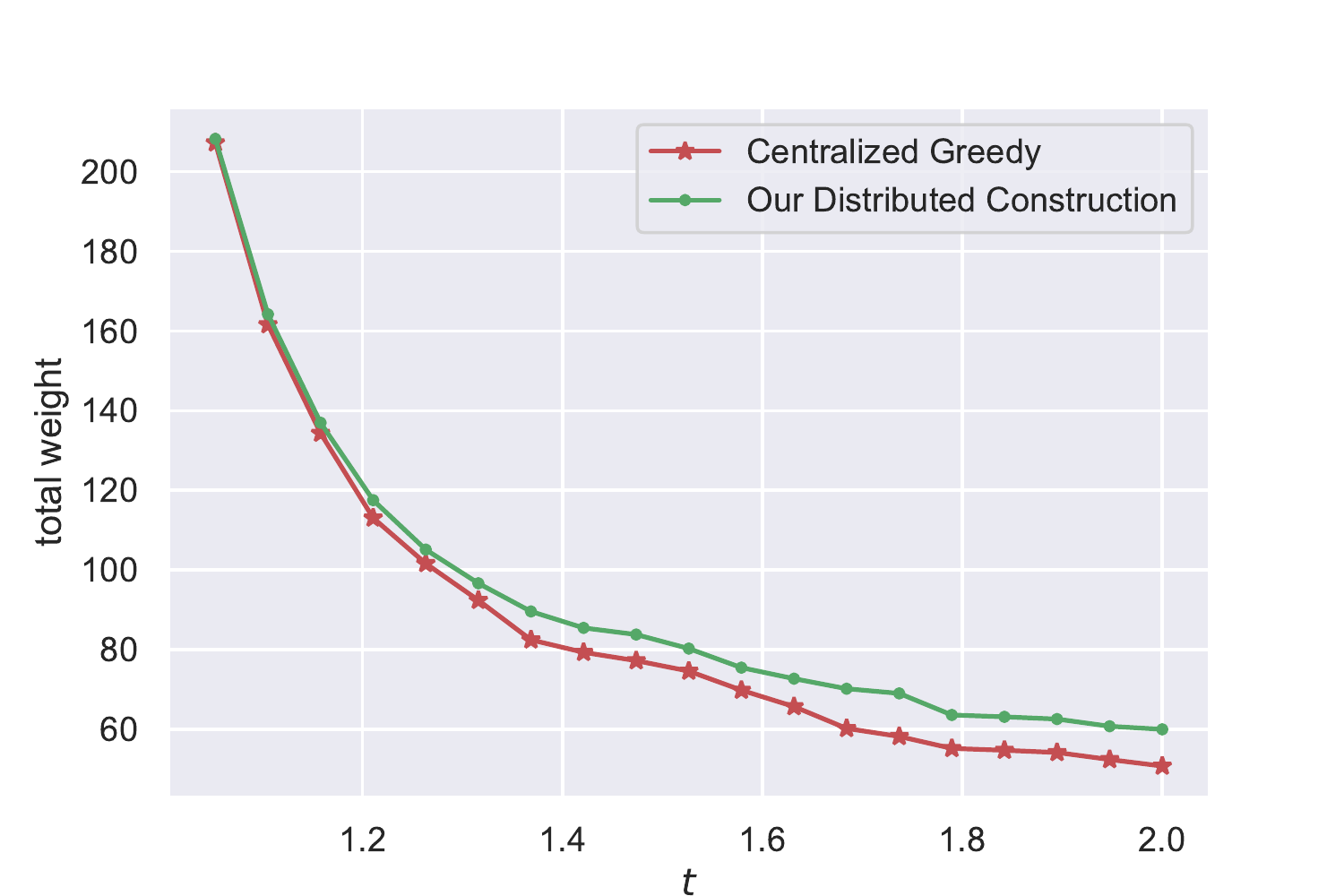}
        \caption{Weight comparison}
    \end{subfigure}
    \hfill
    \caption{Comparisons for a random instance of $100$ points uniformly taken from a $5\times5$ square, for different values of the stretch parameter, $t$.}
    \label{fig:exp-2}
\end{figure}

Next, we define the efficiency of a distributed algorithm with respect to a measure. Let $\mathcal{M}$ be a measure that we would like to minimize, e.g. maximum degree, size, or total weight. Then we define the \emph{efficiency} of a distributed construction with respect to $\mathcal{M}$ to be the ratio
$$\frac{\mathcal{M}(Greedy)}{\mathcal{M}(ALG)}$$
where $ALG$ is the output of the distributed algorithm and $Greedy$ is the output of the centralized greedy on the same point set. We choose the centralized greedy as our base of comparison because it is known to be near-optimal. In the second part of our experiments, we compare the efficiency of our distributed algorithm against the centralized greedy algorithm with respect to the maximum degree, size, and total weight, for 10 random point sets chosen the same way as in the first part. The average of these efficiencies for each measure is reported in Figure \ref{fig:exp-3}.

\begin{figure}[ht]
    \centering
    \includegraphics[trim={0.2cm 0cm 1cm 0.7cm},clip,width=0.4\textwidth]{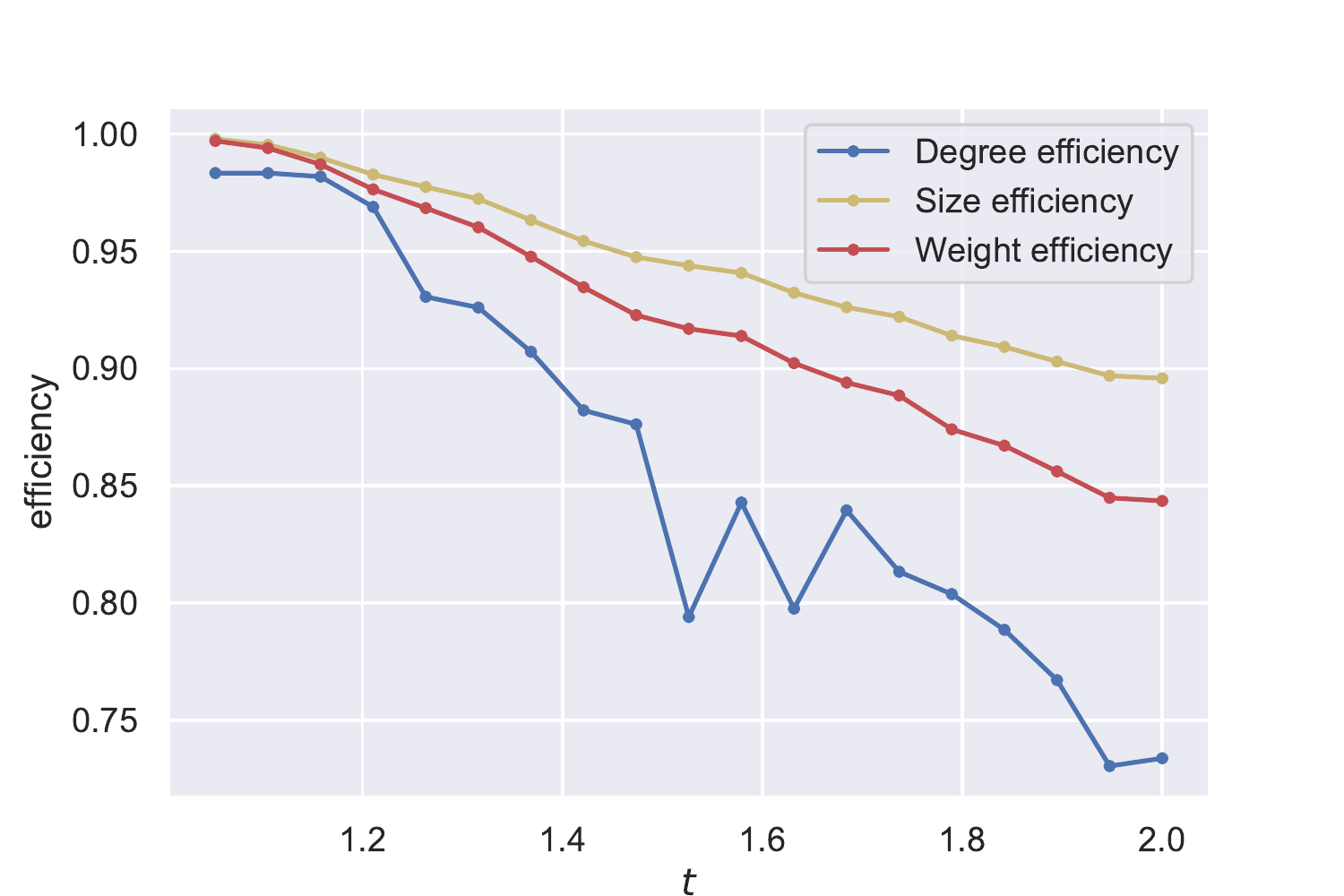}
    \caption{The average efficiencies with respect to maximum degree, size, and total weight, for 10 random instances of $100$ points uniformly taken from a $5\times5$ square, plotted against different stretch parameters, $t$.}
    \label{fig:exp-3}
\end{figure}

We again observe that our distributed algorithm is performing efficiently with respect to the size and total weight. We denote that when $t=2$, the maximum degree of the greedy spanner is 5, so even a single extra edge around any vertex would cause an efficiency below $83\%$. Therefore, even for high values of $t$ our algorithm is performing decently.

\section{Conclusions}
In this paper we resolved an open question from 2006 and we proved the existence of light-weight bounded-degree $(1+\epsilon)$-spanners for unit ball graphs in the spaces of bounded doubling dimension. Moreover, we provided a centralized construction and a distributed construction in the LOCAL model that finds a spanner with these properties. Our distributed construction runs in $\mathcal{O}(\log^* n)$ rounds, where $n$ is the number of vertices in the graph. If a maximal independent set of the unit ball graph is known beforehand, our algorithm runs in constant number of rounds. Next, we showed how to adjust our distributed construction to work in the CONGEST model, without touching its asymptotic round complexity. In this way, we provided the first CONGEST algorithm for finding a light spanner of unit ball graphs.

In addition, we further adjusted these algorithms (in section \ref{sec:euc}) for the case of unit disk graphs in the two dimensional Euclidean plane, and we presented the first centralized and distributed constructions for a light-weight bounded-degree $(1+\epsilon)$-spanner that also has a linear number of edge intersections in total. This can be useful for practical purposes if minimizing the number of edge intersections is a priority. We proved, based on this low-intersection property, that our spanner has sub-linear separators, and a separator hierarchy, and we were able to generalize this result to higher dimensions of Euclidean spaces.

Finally, we ran experiments (in section \ref{sec:exp}) on random point sets in the two dimensional Euclidean plane, to ensure that our theoretical bounds are also supported by enough empirical evidence. Our results show that our construction performs efficiently with respect to the maximum degree, size, and total weight.

\bibliographystyle{plainurl}
\bibliography{reference}

\appendix

\section{Proof of Lemma \ref{lem:int-ES}}%
\label{sec:geometric}

\begin{proof}
We divide the proof into two cases:
\begin{enumerate}
    \item Intersections of $AB$ with significantly larger edges, and
    \item Intersections of $AB$ with edges of within a constant factor length difference.
\end{enumerate}

For each case we prove a constant bound on the number of intersections of $AB$ with those edges. First, we consider the intersections with significantly larger edges. It would be sufficient if we prove a constant bound on the number of $\theta$-parallel large edges, where $\theta$ is a small constant. By $\theta$-parallel (almost-parallel) we mean a set of edges that their pair-wise angle is at most $\theta$. Eppstein and Khodabandeh \cite{eppstein2020edge} defined an end-point ordering of the spanner edges over a baseline, which is selected arbitrarily from the set of almost parallel segments, and they proved the following lemma.

\begin{lemma}
Let $MN$ and $PQ$ be two intersecting segments from a set of $\theta$-parallel spanner segments. Also assume that $\theta < \frac{t-1}{2t}$ where $t$ is the stretch factor of the spanner. Then $MN$ and $PQ$ are endpoint-ordered, i.e. the projection of one of the segments on the baseline of the set cannot be included in the projection of the other one.
\label{lem:interorder}
\end{lemma}

This lemma assumes that $MN$ and $PQ$ intersect each other on an interior point. We can generalize this lemma to a case that they both intersect an arbitrary segment $AB$ on the plane.

\begin{lemma}
Let $MN$ and $PQ$ be two segments chosen from a set of $\theta$-parallel spanner segments that cross an arbitrary segment $AB$. Also assume that $\theta<\frac{t-1}{2(t+1)}$, and $\min(|MN|, |PQ|)\geq\frac{3t(t+1)}{t-1}|AB|$, where $t$ is the spanner parameter. Then $MN$ and $PQ$ are endpoint-ordered.
\label{lem:order}
\end{lemma}
\begin{proof}
The proof goes by contradiction. Let $l$ be an arbitrarily chosen baseline from our set of $\theta$-parallel segments. Without loss of generality suppose that the projections of $P$ and $Q$ on the baseline are both between the projections of $M$ and $N$. We use Lemma \ref{lem:interorder} to show that $MN$ can be shortcut by $PQ$ by a factor of $t$, i.e.
$$t\cdot |MP| + |PQ| + t\cdot |QN| \leq t\cdot |MN|$$
We move $PQ$ by slightly with respect to its length, so that the new segment and $MN$ intersect each other, and then we use Lemma \ref{lem:interorder} for these segments. Let $MN$ and $PQ$ intersect $AB$ at $S$ and $T$, respectively. We move $PQ$ by $\overrightarrow{TS}$ so that it intersects $MN$. Let $P'Q'$ be the result of the movement. The projections of $P'$ and $Q'$ on the baseline may not be between $M$ and $N$. We can extend $MN$ on one side by $|\overrightarrow{TS}|$ to get a new segment $M'N'$, and the property would hold afterwards. Before the movement the projections of $P$ and $Q$ are both between the projections of $M$ and $N$, so after movement at most one of the projections of $P'$ or $Q'$ can be outside of the projections of $M$ and $N$. So extending on one side will be sufficient.

\begin{figure}[t]
\centering
\includegraphics[width=0.4\textwidth,trim=1.5cm 4.5cm 3cm 2.5cm,clip]{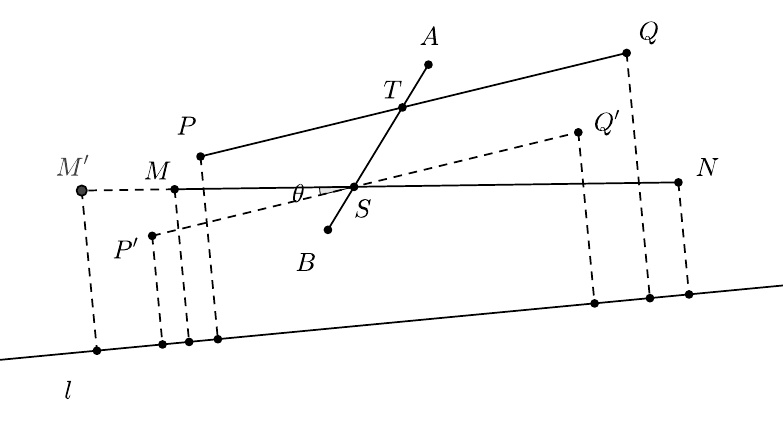}
\caption{Proof of Lemma \ref{lem:order}.}
\label{fig:lemorder}
\end{figure}

Now e can use Lemma \ref{lem:interorder} for $P'Q'$ and $M'N'$. By the assumption $\theta=\frac{t'-1}{2t'}$ where $t'=(t+1)/2$, so Lemma \ref{lem:interorder} implies that,
\begin{equation}
t'\cdot |M'P'| + |P'Q'| + t'\cdot |Q'N'| \leq t'\cdot |M'N'|
\label{sni-e1}
\end{equation}
By the triangle inequality after this movement $MP$ and $NQ$ each decrease by at most $|\overrightarrow{TS}|\leq |AB|$. So,
\begin{equation}
|M'P'| \geq |MP| - |AB|,\; |N'Q'| \geq |NQ| - |AB|
\label{sni-e2}
\end{equation}
Also length of $MN$ will increase by at most $|\overrightarrow{TS}| \leq |AB|$, so
\begin{equation}
|M'N'| \leq |MN| + |AB|
\label{sni-e3}
\end{equation}
The length of $PQ$ does not change. From equations \ref{sni-e1}, \ref{sni-e2}, and \ref{sni-e3},
\begin{align*}
|PQ| = |P'Q'| &\leq t'\cdot(|M'N'| - |M'P'| - |N'Q'|) \\
&\leq \frac{t+1}{2}\cdot(|MN| - |MP| - |NQ| + 3|AB|) \\
&\leq \frac{t+1}{2}\cdot(|MN| - |MP| - |NQ|) + \frac{t+1}{2}\cdot\frac{t-1}{t(t+1)}|PQ|
\end{align*}
So
$$|PQ| \leq t\cdot(|MN|-|MP|-|NQ|)$$
\end{proof}

This implies a total ordering of the set of almost-parallel edges on the given baseline.
\begin{corollary}
Given an arbitrary segment $AB$ one the plane, for a set of sufficiently large almost-parallel (with respect to $AB$) spanner edges that intersect $AB$, the endpoint-ordering is a total ordering.
\end{corollary}

Also, by the endpoint-gap property between the edges of the greedy spanner, they proved the following lower bound on the endpoint distance of two greedy spanner edges,
\begin{lemma}
Let $MN$ and $PQ$ be two $\theta$-parallel spanner segments. The matching endpoints of these two segments cannot be closer than a constant fraction of the length of the smaller segment. More specifically,
$$\min(|MP|, |NQ|) \geq \frac{t-1-2\sin(\theta/2)}{2t}\min(|MN|,|PQ|)$$
\label{lem:lowerbd}
\end{lemma}

Lemma \ref{lem:order} together with lemma \ref{lem:lowerbd} can be used to prove a constant bound on the number of such edges.
\begin{lemma}
\label{thm:longer}
For small $\theta$,
the number of sufficiently large $\theta$-parallel edges of a greedy $t$-spanner that intersect a arbitrary segment $AB$ on the plane is bounded by a constant. More specifically, the length of the segments should be larger than $\frac{3t(t+1)}{t-1}|AB|$.
\label{lem:larger}
\end{lemma}
This proof is similar to theorem 14 of \cite{eppstein2020edge}, except we use lemma \ref{lem:larger} instead of proposition 13.

For edges whose lengths are between $|AB|$ and $\frac{3t(t+1)}{t-1}|AB|$ on the other hand, we only need the endpoint-gap property of the greedy spanner. Similar to \cite{eppstein2020edge} we can show that,
\begin{lemma}
The number of spanner segments $PQ$ that cross an arbitrary segment $AB$ of the plane and that have length between $\alpha\cdot |AB|$ and $\beta\cdot |AB|$ is bounded by a constant.
\label{lem:same-bound}
\end{lemma}
\begin{proof}
We partition the area around $AB$ with squares of edge length $\frac{t-1}{2\sqrt{2}t}\cdot\alpha |AB|$ with edges parallel or perpendicular to $AB$. The area that an endpoint of a crossing segment can lie in is a rectangle of size $(2\beta +1)|AB|$ by $2\beta |AB|$. The total number of cells in this area would be
$$\frac{2\beta(2\beta+1)}{\alpha^2}\cdot\frac{8t^2}{(t-1)^2}$$
And for each crossing segment the pair of cells that contain the two endpoints of the segment is unique. Otherwise two segments, e.g. $MN$ and $PQ$, will have both endpoints at the same pair, which means
\begin{align*}
    \max(|MP|, |NQ|) &< (\sqrt{2})(\frac{t-1}{2\sqrt{2}t}\cdot\alpha |AB|) = \frac{t-1}{2t}\cdot\alpha |AB| \\
    &\leq \frac{t-1}{2t}\min(|MN|, |PQ|)
\end{align*}
which is impossible due to lemma 15 of \cite{eppstein2020edge}. So the total number of such edges would be at most,
$$\left[\frac{2\beta(2\beta+1)}{\alpha^2}\cdot\frac{8t^2}{(t-1)^2}\right]^2$$
\end{proof}

Putting together lemma \ref{lem:larger} and lemma \ref{lem:same-bound} we can deduce lemma \ref{lem:int-ES}.
\end{proof}

\end{document}